\documentclass[acmsmall,natbib,screen=true]{acmart}

\usepackage{mathtools}
\usepackage{fixmath}
\usepackage{graphicx} %
\usepackage{subfigure}
\usepackage{xcolor}
\usepackage{pifont}
\usepackage{color}
\usepackage{siunitx}
\usepackage[inline]{enumitem}
\usepackage{geometry}
\usepackage{booktabs}
\usepackage{acronym}
\usepackage{soul}
\usepackage[skip=1pt]{caption}
\usepackage{adjustbox}
\usepackage{tcolorbox}
\usepackage{multicol}
\usepackage{multirow}
\usepackage{amsmath}
\usepackage{amsfonts}
\usepackage{fontawesome}
\usepackage{acronym}
\usepackage{xspace}
\usepackage[utf8]{inputenc}
\usepackage{microtype}
\usepackage{latexsym}
\usepackage{booktabs,caption}
\usepackage[flushleft]{threeparttable}
\usepackage{makecell}
\usepackage[dvipsnames]{colortbl}
\usepackage[ruled,vlined]{algorithm2e}
\usepackage{verbatim}
\usepackage{arydshln}

\usepackage{listings}

\definecolor{lemon}{HTML}{FDFFCC}

\definecolor{Gainsboro}{rgb}{0.86, 0.86, 0.86}

\definecolor{Gray}{gray}{0.95}
\definecolor{LightCyan}{rgb}{0.88,1,1}
\definecolor{dm-blue-500}{RGB}{0, 69, 177}
\definecolor{dm-purple-500}{RGB}{105,50,230}
\definecolor{dm-red-500}{RGB}{255,122,122}
\definecolor{backred}{RGB}{255, 190, 190}
\definecolor{backblue}{RGB}{220, 230, 250}

\newtcbox{\hlprimarytab}{on line, rounded corners, box align=base, colback=backblue, colframe=white, size=fbox, arc=3pt, before upper=\strut, top=-2pt, bottom=-4pt, left=-2pt, right=-2pt, boxrule=0pt}
\newtcbox{\hlsecondarytab}{on line, box align=base, colback=backred, colframe=white, size=fbox, arc=3pt, before upper=\strut, top=-2pt, bottom=-4pt, left=-2pt, right=-2pt, boxrule=0pt}

\newcommand{\ours}{\textsc{DRO}\xspace}

\newcommand{\todo}[1]{\textcolor{black}{#1}}

\newtheorem{lemma}{Lemma}
\newtheorem{remark}{Remark}[section]

\makeatletter
\newcommand\headernodot{\def\@toclevel{4}%
  \@startsection{paragraph}{4}{\z@}%
  {-.1\baselineskip \@plus -1\p@ \@minus -.1\p@}%
  {-2.5\p@}%
  {\ACM@NRadjust{\bfseries}}}
\newcommand{\header}[1]{\headernodot{#1.}}
\makeatother

\acrodef{rag}{retrieval-augmented generation}
\acrodef{RAG}{retrieval-augmented generation}
\acrodef{LLM}{large language model}

\newcommand{\perm}{\boldsymbol{z}}
\newcommand{\doc}{\boldsymbol{d}}
\newcommand{\elbo}{\text{ELBO}}
\newcommand{\estep}{document permutation estimation\xspace}

\newcommand{\mstep}{re-weighted maximization\xspace}

\acmYear{2025}
\setcopyright{rightsretained}
\acmConference[Conference acronym 'XX]{2025}
\acmDOI{}
\acmISBN{}


\begin{document}

% \input{Sections/reviews}
% \newpage
% \clearpage

\title[Direct Retrieval-augmented Optimization: Synergizing Knowledge Selection and Language Models]{Direct Retrieval-augmented Optimization: \\ Synergizing Knowledge Selection and Language Models}

\author{Zhengliang Shi}
\orcid{0000-0002-9658-4906}
\affiliation{%
  \institution{Shandong University} 
  \city{Qingdao}
  \country{China}
  \postcode{266237}
}
\email{zhengliang.shii@gmail}

\author{Lingyong Yan}
\orcid{0000-0002-6547-1984}
\affiliation{%
  \institution{Baidu Inc.}
    \city{Beijing}
  \country{China}
}
\email{lingyongy@gmail.com}

\author{Weiwei Sun}
\orcid{0000-0002-4817-9500}
\affiliation{%
  \institution{Carnegie Mellon University} 
      \city{Pittsburgh}
  \country{United States}
}
\email{sunweiwei@gmail.com}

\author{Yue Feng}
\orcid{0009-0000-6665-6406}
\affiliation{%
  \institution{University of Birmingham}
    \city{Birmingham}
  \country{UK}
}
\email{y.feng.6@bham.ac.uk}

\author{Pengjie Ren}
\orcid{0000-0003-2964-6422}
\affiliation{%
  \institution{Shandong University} 
    \city{Qingdap}
  \country{China}
}
\email{jay.r@outlook.com}

\author{Xinyu Ma}
\orcid{0000-0002-5511-9370}
\affiliation{%
  \institution{Baidu Inc.}
    \city{Beijing}
  \country{China}
}
\email{xinyuma2016@gmail.com}

\author{Shuaiqiang Wang}
\orcid{0000-0002-9212-1947}
\affiliation{%
  \institution{Baidu Inc.}
    \city{Beijing}
  \country{China}
}
\email{shqiang.wang@gmail.com}

\author{Dawei Yin}
\orcid{0000-0002-0684-6205}
\affiliation{%
  \institution{Baidu Inc.}
    \city{Beijing}
  \country{China}
}
\email{yindawei@acm.org}

\author{Maarten de Rijke}
\orcid{0000-0002-1086-0202}
\affiliation{%
  \institution{University of Amsterdam} 
        \city{Amsterdam}
  \country{Netherland}
}
\email{m.derijke@uva.nl}

\author{Zhaochun Ren}
\orcid{0000-0002-9076-6565}
\affiliation{%
  \institution{Leiden University}
        \city{Leiden}
  \country{Netherland}
}
\authornote{Corresponding author.}
\email{z.ren@liacs.leidenuniv.nl}

\renewcommand{\shortauthors}{Zhengliang Shi et al.}

% \author{
% Zhengliang Shi\textsuperscript{\rm 1},
% Lingyong Yan\textsuperscript{\rm 2}, 
% Weiwei Sun\textsuperscript{\rm 3}, 
% Yue Feng\textsuperscript{\rm 4},
% Pengjie Ren\textsuperscript{\rm 1},
% Xinyu Ma\textsuperscript{\rm 2},\\
% Shuaiqiang Wang\textsuperscript{\rm 2},
% Dawei Yin\textsuperscript{\rm 2}, 
% Maarten de Rijke\textsuperscript{\rm 5}, 
% Zhaochun Ren\textsuperscript{\rm 6*}
% }
% \affiliation{
% \textsuperscript{\rm 1}Shandong University \quad 
% \textsuperscript{\rm 2}Baidu Inc. \quad 
% \textsuperscript{\rm 3}Carnegie Mellon University \quad \textsuperscript{\rm 4}University of Birmingham \quad \\
% \textsuperscript{\rm 5}University of Amsterdam\quad \textsuperscript{\rm 6}Leiden University
% }
% \email{{zhengliang.shii, lingyongy, weiweisun}@gmail.com}
% \email{yindawei@acm.org, m.derijke@uva.nl, z.ren@liacs.leidenuniv.nl}

\renewcommand{\shortauthors}{Shi et al. (SDU, Baidu, CMU, UoB, UvA, Leiden)}

\begin{abstract}
\Acf{RAG} integrates \acp{LLM} with retrievers to access external knowledge, improving the factuality of LLM generation in knowledge-grounded tasks.
To optimize the RAG performance, most previous work independently fine-tunes the retriever to adapt to frozen \acp{LLM} or trains the \acp{LLM} to use documents retrieved by off-the-shelf retrievers, lacking end-to-end training supervision.
Recent work addresses this limitation by jointly training these two components but relies on overly simplifying assumptions of document independence, which has been criticized for being far from real-world scenarios.
Thus, effectively optimizing the overall RAG performance remains a critical challenge.
% \todo{A crucial challenge is synergizing the retrieval and generation process to improve overall performance.}

We propose a \textbf{d}irect \textbf{r}etrieval-augmented \textbf{o}ptimization framework, named \ours, that enables end-to-end training of two key components: 
\begin{enumerate*}[label=(\roman*)]
\item a generative knowledge selection model and 
\item an LLM generator.
\end{enumerate*}
\ours alternates between two phases: 
\begin{enumerate*}[label=(\roman*)]
\item document permutation estimation and 
\item re-weighted maximization, progressively improving RAG components through a variational approach.
\end{enumerate*}
In the estimation step, we treat \textit{document permutation} as a latent variable and directly estimate its distribution from the selection model by applying an importance sampling strategy.
In the maximization step, we calibrate the optimization expectation using importance weights and jointly train the selection model and LLM generator.
Our theoretical analysis reveals that \ours is analogous to policy-gradient methods in reinforcement learning.
Extensive experiments conducted on five datasets illustrate that \ours outperforms the best baseline with 5\%--15\% improvements in EM and F1.
We also provide in-depth experiments to qualitatively analyze the stability, convergence, and variance of \ours.\footnote{Code is available on \href{https://github.com/mangopy/direct-rag-learning}{\faGithub~GitHub}.}
\end{abstract}




\maketitle

\acresetall

\vspace{-2mm}
\section{Introduction}\label{sec:intro}

Large language models (LLMs) have shown remarkable text generation abilities; however, they often provide factually incorrect content~\cite{zhang2023siren, dhuliawala2023chain,su2024mitigating} due to the hallucination~\cite{huang2023survey} or out-of-date information~\cite{fan2024survey}.
To mitigate these limitations, \acf{RAG} is proposed to integrate external retrievers with LLMs, which enables the model to access extensive corpora and retrieve relevant documents for references, thereby enhancing factuality. 
By integrating the retriever with LLMs, RAG has shown superior performance in knowledge-intensive tasks such as question answering~\cite{wang2024domainrag, shi2024generate} and conversational information seeking~\cite{dinan2018wizard,li2024corpuslm,yang2024rag}.

Following the most widely used architecture~\cite{Lewis2020RetrievalAugmentedGF,gao2023retrieval,fan2024survey}, \ac{RAG} typically includes two components to answer an input query:
\begin{enumerate*}[label=(\roman*)]
    \item \textit{knowledge selection}, where retrieval and re-ranking models select target documents,
    \item \textit{answer generation}, where an LLM generator generates correct answers conditioned on the selected documents.
\end{enumerate*}
To enhance coverage and improve answer quality, \ac{RAG} models often provide multiple retrieved documents as input to the generator. The interrelationships among these documents are crucial for final performance~\cite{liu2024lost, min2020ambigqa,hofstatter2023fid,zamani2024stochastic}. We refer to a specific selection of retrieved documents as a \textit{document permutation}.

\begin{figure}[!t]
    \centering
    \includegraphics[width=0.7\linewidth]{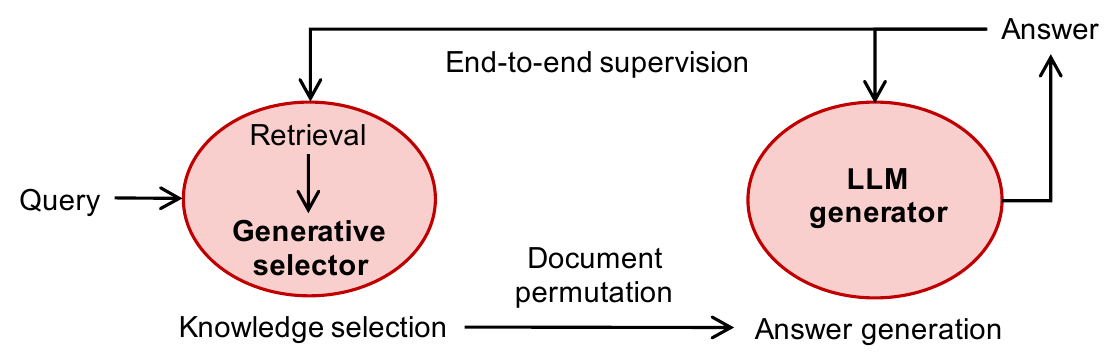}
    \caption{Overview of \ours objective. The selection model  directly estimate a document permutation for the generator to predict an answer, with both components trained jointly.} 
\end{figure}\label{fig:intro}

\header{Improving RAG performance}
To optimize \ac{RAG} performance, some studies improve knowledge accuracy by fine-tuning the retrieval or ranking model with relevance criteria~\cite{moon1996expectation,shi2023replug, pradeep2023rankzephyr}. Others enhance the robustness of LLMs against irrelevant content through supervised fine-tuning~\cite{yu2024rankrag, fang2024enhancing} or in-context learning~\cite{shi2024generate}, teaching them to summarize key points from retrieved documents. However, these approaches optimize either the \textit{selection} or the \textit{generation} component separately while neglecting a dual enhancement, which may lead to sub-optimal overall performance~\cite{lin2023ra,erqrag}.
 
To address the above limitation, some recent studies train both the retriever and the LLM generator~\cite{Lewis2020RetrievalAugmentedGF, eerr, li2024rag}. However, to avoid optimizing over the complex permutation distribution of inter-related documents, most of them simplify the retriever outputs into point-wise documents rather than the overall permutations~\cite{zamani2024stochastic}. Specifically, they first retrieve the top-k documents using a dense retriever, independently feed each document into downstream LLMs for training, and train a \textit{point-wise dense retriever} to assign a higher similarity to the document that leads to better downstream outcomes. \textit{However, this is far from reality, as RAG models often gather useful content from multiple documents}~\cite{zamani2024stochastic, wang2024richrag}, which suffers from a pronounced limitation in complex scenarios, such as multi-hop QA. Besides, training a dense retriever requires the frequent updating of the document index~\cite{guu2020retrieval}, which is hard to operate and potentially non-compatible with established retrieval applications. Thus, a natural question is: \textit{How to synergize 
\begin{enumerate*}[label=(\roman*)]
    \item knowledge selection; and
    \item answer generation
\end{enumerate*}
to optimize holistic RAG performance?}

In this paper, we propose \ours, a direct retrieval-augmented optimization method that synergizes (i) a list-wise selection model (\textit{aka}, selector) to generate target document permutations, and (ii) an LLM generator to predict the answer, enabling end-to-end improvement. As shown in Figure~\ref{fig:intro}, the core idea is to directly treat the \textit{document permutation} as a latent variable and estimate its distribution to maximize the log-likelihood of question answering. To achieve this, \ours iteratively alternates between two steps: 
\begin{enumerate*}[label=(\roman*)]
\item \estep and 
\item \mstep 
within the Expectation-Maximization principle~\cite{moon1996expectation}.
\end{enumerate*}

In the \textit{permutation estimation} step, we first define an ideal posterior distribution of document permutation inspired by the classic expectation-maximization algorithm~\cite{moon1996expectation}, introducing a tractable evidence lower bound (ELBO) for the log-likelihood objective.
Considering exactly computing the posterior distribution is typically impractical, we employ an importance sampling strategy~\cite{kloek1978bayesian,elvira2021advances} to directly estimate the document permutation distribution by sampling from the selection model.
Specifically, for each input query, we first recall relevant documents using an off-the-shelf dense retriever, filtering documents with no semantic relevance to narrow down the candidate documents. 
The generative selector then selects a subset of documents by generating their document identifiers in an auto-regressive manner (e.g., \texttt{"[1] > [2] > [3]"}).

In the \textit{re-weighted maximization} step, we optimize the ELBO constructed in the estimation step, thereby improving the overall log-likelihood for question answering.
We first re-weight the collected samples using importance weights to calibrate the bias introduced by sampling shifting from the importance sampling strategy.
Then, we jointly optimize the selection model and the LLM generator by maximizing this re-weighted expectation, where both two components are trained with end-to-end supervision.
By alternating the estimation and maximization steps, \ours progressively improves the holistic RAG performance.

\header{Theoretical analysis}
\ours differs from prior work such as~\cite{li2024rag, eerr, Lewis2020RetrievalAugmentedGF} by:
\begin{enumerate*}[label=(\roman*)]
    \item enabling \textit{end-to-end optimization} for both knowledge selection and answer generation, rather than optimizing individual processes; and
    \item \textit{directly estimating} the distribution of document permutations for optimization, relaxing the assumption of independent top-$k$ marginalization posed in prior works;
    \item \todo{\textit{iteratively aligning} of selection and generation models, which achieves a consistent improvement until convergence}. 
\end{enumerate*}
To investigate the advantages of \ours, we provide a theoretical analysis of the learning objective and optimization process within our framework. We prove that \ours shares similarities with policy-based reinforcement learning approaches. In \ours, the selection module improves by reinforcing document permutations that enhance generation performance, while the generator, in turn, benefits from improved document permutations, creating a synergistic loop that optimizes the entire RAG process. 
Additionally, we provide theoretical analysis about the training convergence and stability of \ours.
We  reveal that importance sampling with normalized weights can guarantees variance reduction and non-decreasing ELBO across iterations.

\header{Experiments}
We conduct extensive experiments across a wide range of datasets, including Natural Questions~\cite{kwiatkowski-etal-2019-natural}, 
HotpotQA~\cite{yang2018hotpotqa}, 
2WikiMultihopQA~\cite{xanh2020_2wikimultihop}, MusiQue~\cite{trivedi2021musique}, and
Wizard-of-Wikipedia~\cite{dinan2018wizard}.
The results show that the proposed \ours outperforms best baselines with 5--15\% improvement in EM and F1 metrics.
Additionally, the selection model trained using our method achieves an average precision improvement of 17.78\% in identifying target documents.
We further conduct fine-grained analyses to examine the variance, convergence, and stability of the \ours during the training process. 
We observe substantial variance decay with increased sampling size, and consistent  improvements (e.g., F1 score) over iterations, indicating stable and convergent optimization.
These findings verify that \ours achieves not only strong performance but also robust training dynamics across datasets.

\header{Contributions} 
The main contributions of this paper are:
\begin{enumerate*}[label=(\roman*)]
\item We propose \ours, a direct retrieval-augmented optimization method that treats the document permutation as a latent variable to enable an end-to-end improvement;
\item We provide theoretical analysis for the learning objective of the proposed method and demonstrate its convergence and training stability; and 
\item extensive experiments conducted on five datasets show the improvement of our method, e.g., 5\%--15\% improvement compared with state-of-the-art baselines.
\end{enumerate*}

\section{Related work}

% \header{Retrieval-augmented generation}
\subsection{Retrieval-augmented generation}
Retrieval-augmented generation (RAG) aims to integrate external knowledge into LLMs, improving their factuality~\cite{gao2023retrieval}.
Given an input query, the first process of RAG is to select relevant knowledge, which is typically done by retrieval~\cite{karpukhin2020dense} or ranking model~\cite{moon1996expectation}.
Subsequently, an LLM generator incorporates these candidate documents to generate an answer.
An active research question is how to improve the overall RAG performance.
Some studies aim to improve the knowledge accuracy~\cite{sachan-etal-2021-end, dong2024understand}, such as fine-tuning an answer-aware dense retriever~\cite{shi2023replug, sachan2022improving} or introducing additional modules for document filtering~\cite{wang2023learning, xu2023recomp}.
Other work alternatively enhances the robustness of LLMs to irrelevant content, enabling LLMs to adaptively extract supporting facts from the retrieved documents~\cite{zhu2024atm, yoran2023making}.
However, these methods either optimize the retrieval or the generation process without dual enhancement, potentially leading to sub-optimal performance~\cite{lin2023ra}.
Although existing work proposes the end-to-end training paradigm,  they overly simplify a marginalization optimization through \textit{independent top-k approximation}~\cite{zamani2024stochastic, sachan-etal-2021-end}, where they simply feed top-k documents into downstream LLMs one-by-one and re-score their relevance to optimize the retriever~\cite{Lewis2020RetrievalAugmentedGF, lin2023ra}.
This has been criticized far from the practical scenarios as the RAG system typically consumes multiple documents~\cite{zamani2024stochastic}, \todo{while exhaustively enumerating all possible document permutations is cost-intensive and typically infeasible in practice.}
In this work, we propose \ours, which directly treats the document permutation as a latent variable and estimates its distribution for optimization.

% \header{Knowledge Selection for RAG}
\subsection{Knowledge Selection for RAG}
In RAG, the \textit{knowledge selection} process aims to select target documents that can maximize LLM generation performance~\cite{guo2020deep,sachan2022improving}. To achieve this, prior work typically trains point-wise rankers (e.g., MonoT5~\cite{moon1996expectation}, BGE~\cite{bge_embedding}) on conventional retrieval benchmarks (e.g., MS-MARCO~\cite{bajaj2016ms}) and separately judges the relevance of each document to the input query. 
In contrast, our method applies a generative list-wise selection model~\cite{sun2023chatgpt}, which selects target documents for the input query by generating corresponding document identifiers auto-regressively~\cite{pradeep2023rankzephyr, pradeep2023rankvicuna}. Compared with point-wise ranking, our list-wise selection enables the comparison between multiple documents~\cite{wang2024richrag, ma2023zero}, which can be inherently used to estimate the permutation distribution in our framework. Additionally, unlike previous ranking models trained with semantic relevance criteria~\cite{rao2019bridging}, our selection model is jointly trained with the generator to maximize end-to-end RAG performance.

% \header{Variational approach for optimization}
\subsection{Variational approach for optimization}
The variational approach has been widely applied in unsupervised scenarios and optimization involving latent variables~\cite{sordoni2024joint,sun2020approach}, such as GLEM in graph learning~\cite{zhao2022learning} and SSDM in speech modeling~\cite{lian2024ssdm}.
As a general principle, the variational approach such as the Expectation-Maximization algorithm (EM)~\cite{moon1996expectation}, alternates between the \textit{Expectation} step to compute the posterior distribution and the \textit{Maximization} step to update the model parameter, optimizing the marginalization progressively.
Inspired by this principle, in our work, we treat the document permutation as a latent variable.
Besides, we directly estimate the permutation distribution using a list-wise selection model through the importance sampling~\cite{kloek1978bayesian}.
This strategy diverges from the standard EM algorithm by avoiding the exact computation of the posterior distribution, a process that is impractical due to the large-scale document permutation space.

\section{Preliminaries}

\subsection{Task Definition}
Given an input query $x$, the task of retrieval-augmented generation (RAG) typically consists of two processes:
\begin{enumerate*}[label=(\roman*)]
    \item \textit{knowledge selection} to acquire a set of relevant documents; and
    \item \textit{answer generation} to generate a correct answer $y$ by referring to the acquired documents.
\end{enumerate*}
In the proposed \ours, our first process starts by using an off-the-shelf retrieval model (e.g., ColBERT~\cite{colbertv2}) to recall relevant documents $\doc = \{d_1, d_2, \ldots, d_{|\doc|}\}$ through semantic matching, thereby filtering the irrelevant contents to narrow down the candidates. Then, a generative re-ranking model $\theta_s$ reads these documents and selects a document permutation $\perm$ by generating the corresponding document identifiers such as \texttt{[1] > [2] > \dots > [k]}. Subsequently, an LLM generator $\theta_g$ predicts the answer $y$ based on the input query $x$ and selected documents $\doc_{\perm}$, which can be formulated as $p(y \mid x, \doc_{\perm}; \theta_g) = \prod\nolimits_{t=1}^{|y|} p(y_{t} \mid y_{<t}, x, \doc_{\perm}; \theta_g)$.

\subsection{Generative Knowledge Selection for RAG}
The generative selection model identifies the target documents in a list-wise manner, directly taking the query $x$ and the candidate documents $\doc$ as input, and generates an ordered sequence of document \textit{identifiers} auto-regressively:
\begin{equation}\label{eq:rank}
   p(\perm \mid x ;\theta_s) = \prod\nolimits_{t=1}^{K} p(\perm_{t} \mid \perm_{<t}, x, \doc; \theta_s).
\end{equation}
Here, $\perm=\{\perm_t; t \in [K]\}$ represents the document permutation consisting of $K$ tokens, where each token corresponds to a document identifier (e.g., \texttt{[1]} or \texttt{[2]}). By mapping the document permutations back to the original documents, we obtain $K$ selected documents, denoted as $\doc_{\perm}=\{\doc_{\perm_t} \mid t \in [K]\}$. Compared with traditional point-wise ranking, which assigns each individual document a relevance score to the query, we use this generative selection model for two reasons:
\begin{enumerate*}[label=(\roman*)]
    \item it inherently enables the comparison of multiple candidate documents through the attention mechanism~\cite{vaswani2017attention,Sun2023GenerativeKS}; and 
    \item it allows the direct modeling of document permutations, ordered lists of documents that capture their interrelationships by autoregressively generating the docid list~\citep{Min2021JointPR}.
\end{enumerate*}
\section{Direct RAG optimization}\label{sec:method}

\begin{figure*}[!t]
        \centering
	\includegraphics[width=
 \linewidth]{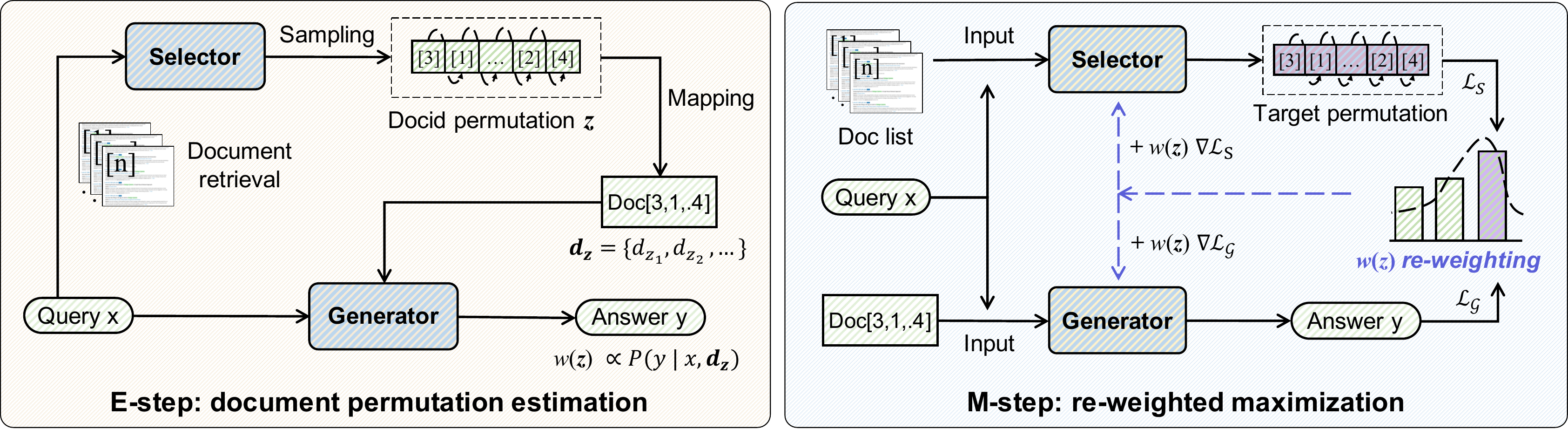}
        \caption{The overall framework for \ours alternates between the (i) E-step: document permutation estimation (Section~\ref{sec:estimation}); and (ii) M-step: re-weighted maximization (Section~\ref{sec:maximization}) to progressively optimize the holistic RAG performance.}\label{fig:method}
\end{figure*}

In this work, the proposed \ours method improves the holistic performance of RAG by synergizing (i) knowledge selection and (ii) answer generation processes. To achieve this, \ours\ enables the end-to-end training of 
\begin{enumerate*}[label=(\roman*)]
    \item a list-wise generative selector with parameters $\theta_s$ to estimate the target document permutation for the input query $x$; and
    \item an LLM generator with parameters $\theta_g$ to generate accurate answers.
\end{enumerate*} 
For simplicity, we use $\theta = (\theta_s, \theta_g)$ to represent the overall tunable parameters of \ours\ throughout the paper.
\todo{Below, we first derive the learning objective of \ours and then relate it to how to achieve the end-to-end optimization of the selection and generation model for such an objective.}

\subsection{Deriving the \ours objective}
\todo{As pointed out by prior work, generating a correct answer grounded on reference knowledge is identified as one of the most crucial goals of RAG tasks~\cite{Lewis2020RetrievalAugmentedGF,gao2023retrieval,erqrag}.}
In \ours, we define the optimization objective as maximizing the marginal log-likelihood of generating the correct answer $y$ to the input query $x$ using external documents $\doc_{\perm}$.
This can be formulated as:
\begin{equation}
    \begin{aligned}
        \log p(y|x; \theta)  &=  \log \sum_{\perm} p(y, \doc_{\perm} | x; \theta).
    \end{aligned}\label{eq:margin}
\end{equation}
Since summing over all possible $(y, \perm)$ is typically intractable, we employ a variational approach to construct a tractable lower bound for $\log p(y | x;\theta)$ that we then maximize. 
Specifically, we introduce a variational distribution $q(\perm|x)$ and apply \textit{Jensen's inequality} to $\log \sum_{\perm} p(y, \doc_{\perm} | x; \theta)$:
\begin{equation}
    \begin{aligned}\label{eq:decomp}
     \log \sum_{\perm}q(\perm | x)\frac{p(y, \doc_{\perm} | x; \theta)}{q(\perm | x)} & \geq \sum_{\perm} q(\perm | x) \log \frac{p(y, \doc_{\perm} | x; \theta)}{q(\perm | x)} \\
        &= \underbrace{\mathbb{E}_{\perm \sim q(\perm | x)} \left[ \log \frac{p(y, \doc_{\perm} | x; \theta)}{q(\perm | x)}\right].}_{\text{\textbf{E}vidence \textbf{L}ower \textbf{BO}und ($\elbo(q, \theta)$)}}
    \end{aligned}
\end{equation}
Here, we define the $\mathbb{E}_{\perm \sim q(\perm \mid x)} \left[ \log \frac{p(y, \doc_{\perm} \mid x; \theta)}{q(\perm \mid x)} \right]$ as an evidence lower bound $(\elbo)$ for $\log p(y|x; \theta)$. Based on Jensen's inequality, the $\elbo(q,\theta) = \log p(y|x; \theta)$ if and only if when the $\frac{p(y, \doc_{\perm} | x; \theta)}{q(\perm | x)} \equiv c$, where $c$ is a constant. 

From Eq.~\eqref{eq:decomp}, we can derive the \ours objective as a progressive optimization process, which includes: (i) estimating the distribution of document permutation $\perm$ to achieve $\log p(y | x; \theta) \approx \elbo$, and (ii) maximizing the $\elbo$ to improve $\log p(y | x;\theta)$.

\header{Expectation-Maximization for \ours training}
To \todo{optimize} the \ours objective, we adopt the expectation-maximization algorithm with an importance sample strategy.
In more detail, we start by demonstrating the condition for the alignment in Eq.~\eqref{eq:decomp}, which is formulated in the following lemma.

\begin{lemma}\label{lemma:elbo}
For $\elbo(q,\theta) = \log p(y|x; \theta)$, there exists a variational distribution $q(z|x)$ such that $q(\perm | x) = p(\perm | x, y;\theta)$.
\end{lemma}
\begin{proof}
    This lemma can be proved by considering the case where the importance weight $\frac{p(y, \perm \mid x; \theta)}{q(\perm \mid x)}$ is constant, denoted as $c$, across all $\perm$. This is formulated as below:
\begin{equation*}
    \begin{aligned}
     &  \frac{p(y, \perm \mid x; \theta)}{q(\perm \mid x)} \equiv c \label{eq:constant1} \quad \forall x, y
    \end{aligned}
\end{equation*}
Then, we can sum both sides over $\perm$ and obtain: $\sum_{\perm} p(y, \perm \mid x; \theta) = c \sum_{\perm} q(\perm|x) \equiv c$.
Here $q(\perm \mid x)$ is a probability distribution over $\perm$, i.e., it sums to 1. Therefore, we solve for $q(\perm \mid x)$ as:
\begin{equation}
    q(\perm \mid x) = \frac{p(y, \perm \mid x; \theta)}{\sum_{\perm} p(y, \perm \mid x; \theta)} = p(\perm \mid x, y; \theta)
\end{equation}
Therefore, the variational distribution $q(z \mid x)$ that matches the true posterior $p(z \mid x, y; \theta)$ achieves the exact $\elbo=\mathbb{E}_{\perm \sim q(\perm \mid x)} \left[ \log \frac{p(y, \doc_{\perm} \mid x; \theta)}{q(\perm \mid x)} \right]$.
\end{proof}

% This Lemma~\ref{lemma:elbo} can be proved as follows:
% \begin{equation*}
%     \begin{aligned}
%      &  \frac{p(y, \perm \mid x; \theta)}{q(\perm|x)} \equiv c \label{eq:constant1} \\
%       \Rightarrow & \sum_{\perm} p(y, \perm \mid x; \theta) = c \sum_{\perm} q(\perm|x) \equiv c \\
%       \Rightarrow & \,\, q(\perm|x) = \frac{p(y, \perm \mid x; \theta)}{\sum_{\perm} p(y, \perm \mid x; \theta)} = p(\perm \mid x, y; \theta).
%     \end{aligned}
% \end{equation*}

The Lemma~\ref{lemma:elbo} shows an intuitive solution to achieve $\log p(y|x;\theta) \approx \elbo(q, \theta)$ by exactly computing the posterior distribution $p(\perm \mid x,y; \theta)$ for latent document permutation $z$.
However, it is often impractical due to the large-scale permutation space.
To address this challenge, we use an importance sampling strategy, where the $q(\perm|x)$ is directly set to $p(\perm \mid x; \theta_s)$.
Consequently, the training of \ours is achieved within an expectation-maximization principle, including:
\begin{enumerate}[label=(\roman*)]
    \item E-step: \textit{document permutation estimation} (Section~\ref{sec:estimation}) to estimate the distribution of document permutations by sampling from the selection model; and
    \item M-step: \textit{re-weighted maximization} (Section~\ref{sec:maximization}) to jointly optimize the selection model and generator using the importance-weighted samples.
\end{enumerate}
By iteratively alternating these two step, \ours progressively improves the holistic RAG objective $\log p(y|x; \theta)$.

\subsection{E-step: Document permutation estimation}\label{sec:estimation}
\todo{This step aims to estimate the distribution of document permutations.}
Specifically, at the $t$th iteration, we first assume $q(\perm \mid x) \approx p(\perm \mid y, x; \theta^{t})$ and transform the $\elbo(\theta, \theta^{t})$ in Eq.~\eqref{eq:decomp} into:
\begin{equation}\label{eq:elbo-before}
        \begin{aligned}
        &\mathbb{E}_{\perm \sim p(\perm \mid x, y ; \theta^{t})}\left[\log p(y, \perm \mid x; \theta) - \log p(\perm \mid x, y ; \theta^{t})
        \right] \\
        &{}=\mathbb{E}_{\perm \sim p(\perm \mid x, y ; \theta^{t})}\left[\log p(y, \perm \mid x; \theta)\right] +  \mathcal{H}\left(p(\perm | x, y ; \theta^{t}\right),
        \end{aligned}
\end{equation}
where the $\perm$ is ideally sampled from the posterior distribution $p(\perm \mid x, y ; \theta^{t})$ to compute the expectation.
The $\mathcal{H}(p(\perm|x,\theta^{t}))$ indicates the entropy of $p(\perm|x,\theta^{t})$, which is independent to $\theta$ and can be viewed as a constant.
Since the posterior distribution $p(\perm | y, x;\theta)$ is intractable, we alternatively adopt an importance sampling strategy to directly sample the document permutation $\perm$ from our selection model $\theta_s$ via $\perm \sim p(\perm \mid x; \theta_s^{t})$.
To correct the bias introduced by the sampling shifting, we also employ an importance weight to calibrate the expectation.
Formally, it is presented as follows:
\begin{equation}\label{eq:elbo-is}
    \begin{aligned}
        \elbo(\theta, \theta^{t}) &= \mathbb{E}_{\perm \sim p(\perm | x; \theta_s^{t})}\left[w(\perm) \log p(y, \perm | x; \theta) \right] + \text{constant} \\
        w(\perm) &= \frac{p(\perm \mid x, y ; \theta^{t})}{p(\perm \mid x; \theta_s^{t})}.
    \end{aligned}
\end{equation}
Here we denote the $w(\perm)$ as the \textit{importance weight}.
According to Bayes' theorem, where $p(\perm \mid y, x; \theta^{t})  \propto p(\perm \mid x; \theta_s^{t})  \times p(y \mid \doc_{\perm}, x; \theta_g^{t})$, we can then simplify the $w(\perm)$ as:
\begin{equation}
    w(\perm) \propto \frac{ p(\perm \mid x; \theta_{s}^{t})  \times p(y \mid \doc_{\perm}, x; \theta_{g}^{t})}{p(\perm \mid x; \theta_s^{t})} = p(y \mid x, \doc_{\perm}; \theta_g^{t}).
\end{equation}
Intuitively, the weight $w(\perm)$ reflects the utility of a set of documents $\doc_{\perm}$ for the LLM generator $\theta_g$ in generating ground-truth answers.
By sampling from the generative selection model $\theta_s$ ($\perm \sim p(\perm \mid x; \theta_s)$) and calibrating the expectation with importance weights, we directly estimate the distribution of the latent variable $\perm$ (i.e., the document permutation) for unbiased optimization, without explicitly computing the posterior distribution $p(\perm \mid x, y ; \theta^{t})$.

\RestyleAlgo{ruled}
\begin{algorithm}[t]
\KwIn{selection model $\theta_{s}$; LLM generator $\theta_{g}$; training iteration number $N$; training data $\mathcal{T}$. }
\For{$t = 1$ to $N$}{
    \textcolor{dm-blue-500}{\texttt{// Permutation Estimation (E-step)}} \\
    \For{each input query $x$ in training set}{
        Sample document permutations: $\perm \sim p(\perm \mid x; \theta_{s}^{t})$\\
        Compute importance weight: $w(\perm) = p(y \mid x, \perm ; \theta_{g}^{t})$\\
    }
    \textcolor{dm-purple-500}{\texttt{// Re-weighting  (M-step)}} \\
$\mathcal{L_{S}}(x;\theta_{s}) := - \mathbb{E}_{\perm \sim p(\perm \mid x; \theta_{s}^{t})}\left[w(\perm) \log p(\perm \mid x; \theta_{s})\right]$ \\
$\mathcal{L_{G}}(x;\theta_{g}) := - \mathbb{E}_{\perm \sim p(\perm \mid x; \theta_{s}^{t})}\left[w(\perm) \log p(y \mid x, \doc_{\perm}; \theta_{g}) \right] $\\
    \textcolor{dm-purple-500}{\texttt{// Re-weighted Maximization  (M-step)}} \\
$\theta_{s}^{t+1} = \arg \max_{\theta_s} \mathbb{E}_{(x, y) \sim \mathcal{T}} \left[ \mathcal{-L_{S}}(x;\theta_{s}) \right]$ \\
$\theta_{g}^{t+1} = \arg \max_{\theta_{g}} \mathbb{E}_{(x, y) \sim \mathcal{T}}\left[-\mathcal{L_{G}}(x;\theta_{g}) \right] $\\
    \If{no improvement on validation set}{
        Stop training Maximization \textcolor{dm-red-500}{\texttt{// Early Stop}} \\
    }
}
\KwOut{$(\theta_s, \theta_g)$}
\caption{The algorithm for \ours, which alternates between \texttt{Estimation} and \texttt{Maximization} steps to progressively improve the overall RAG performance.}
\label{algo:opt}
\end{algorithm}

\subsection{M-step: Re-weighted maximization}\label{sec:maximization}
After the permutation estimation step, the maximization step aims to update the tunable parameter $\theta$ in the RAG system.
Formally, the optimization objective is defined as: 
\begin{align}
\theta^{t+1} = \arg\max_{\theta} \mathbb{E}_{\perm \sim p(\perm  \mid x; \theta_s^{t})}\left[ w(\perm) \log p(y, \perm \mid x; \theta) \right].
\label{eq:update}
\end{align}
Here, the $\theta = (\theta_s, \theta_g)$ denotes the tunable parameters of knowledge selection model $\theta_s$ and LLM generator $\theta_g$.
Based on Bayes' theorem, we have $p(y, \perm \mid x; \theta) = p(\perm \mid x;\theta_s) \cdot p(y \mid x, \doc_{\perm}; \theta_g)$.
Then, the $\elbo(\theta, \theta^{t})$ can be rewritten as a decomposition of two parts:
\begin{equation}
\mathbb{E}_{\perm  \sim p(\perm  \mid x; \theta_s^{t})}  \bigg[\underbrace{w(\perm) \log p(\perm \mid x; \theta_s)}_{\text{\textit{Learning to select}}} + \underbrace{w(\perm ) \log p(y \mid x, \doc_{\perm}; \theta_g)}_{\text{\textit{Learning to generate}}}\bigg].
\label{eq:two-part}
\end{equation}
Thus, we derive two loss functions, namely (i) selection optimization $\mathcal{L_S}$ and (ii) generation optimization $\mathcal{L_G}$:
\begin{align}
\mathcal{L_S}:= & - \mathbb{E}_{\perm  \sim p(\perm  \mid x; \theta_s^{t})} \left[w(\perm) \log p(\perm \mid x; \theta_s)\right] \label{eq:rank-opt},\\
\mathcal{L_G} := & - \mathbb{E}_{\perm  \sim p(\perm  \mid x; \theta_s^{t})} \left[w(\perm) \log p(y \mid x, \doc_{\perm} ; \theta_g)\right]. \label{eq:gen-opt}
\end{align}

\header{Learning to select}
The function $\mathcal{L_S}$ optimizes the selection model $\theta_s$ for document permutation generation.
Since $\mathcal{L_S}$ is weighted by $w(\perm)$ and $w(\perm) \propto p(y|x, \doc_{\perm}; \theta_s^{t})$, the model $\theta_s$ learns to generate the document permutation that maximizes the end-to-end performance.
Based on Eq.~\eqref{eq:rank-opt}, the gradient of $\mathcal{L_S}$ with respect to $\theta_s$ is:
\begin{equation}
\begin{aligned}
\nabla_{\theta_s} \mathcal{L_S} = - \mathbb{E}_{\perm \sim p(\perm \mid x; \theta_s^t)} \left[ w(\perm) \nabla_{\theta_s} \log p(\perm \mid x; \theta_s) \right].
\label{eq:gradient-rank}
\end{aligned}
\end{equation}

\header{Learning to generate}
The loss function for $\theta_g$ is defined as $\mathcal{L_G} := - \mathbb{E}_{\perm  \sim p(\perm  \mid x; \theta_s^{t})} \left[w(\perm) \log p(y \mid x, \perm ; \theta_g)\right]$, training the LLM generator to understand the selected documents and generate correct answers.
The gradient of $\mathcal{L_{G}}$ with respect to $\theta_g$ can be formulated as:
\begin{equation}
\begin{aligned}
\nabla_{\theta_g} \mathcal{L_G} = - \mathbb{E}_{\perm \sim p(\perm \mid x; \theta_s^t)} \left[ w(\perm) \nabla_{\theta_g} \log p(y \mid x, \doc_{\perm}; \theta_g) \right].
\label{eq:gradient-gen}
\end{aligned}
\end{equation}

\subsection{Upshot: Explanation for \ours with pseudo algorithm}
To provide a more intuitive explanation of the overall optimization process, we present a pseudo algorithm in Algorithm~\ref{algo:opt}.
In \ours, the document permutation distribution is first estimated using the selection model $\theta_s$ (E-step).
Subsequently, the model $\theta_s$ learns to select documents that maximize the generation performance while the generator $\theta_s$, in turn, further learns to leverage the selected documents (M-step).
These two steps are iteratively alternated, enabling the progressive improvement of the holistic RAG system.

\section{Theoretical analysis}\label{sec:theoretical}

To further interpret \ours, we offer a theoretical analysis of its advantages and training stability, and prove its convergence.

\subsection{What do the selector and generator learn?}\label{sec:em-rl}

The format of the optimization gradient in Eq.~\eqref{eq:gradient-rank} and Eq.~\eqref{eq:gradient-gen} is similar to classic policy gradient  approaches~\cite{sutton1999policy,ahmadian2024back,zhao2024policy} used in the reinforcement learning (RL).
Generally, given an input task $x$ and the action $\tau$ generated by the policy model $\pi$, the policy gradient objective $\mathcal{J}(\pi)$ to maximize an reward function $r(\tau)$ can be presented as:
\begin{equation}
\begin{aligned}
    \mathcal{J}(\pi) &= \mathbb{E}_{\tau \sim p(\tau \mid x; \pi) } \left[ r(\tau)\right], \\
   \nabla_{\pi} \mathcal{J}(\pi) &= \mathbb{E}_{\tau \sim p(\tau \mid x; \pi)}\left[ r(\tau) \cdot  \nabla_{\pi} \log p(\tau \mid x;\pi) \right].
   \label{eq:rl}
\end{aligned}
\end{equation}
Below, we make several comparisons to relate RL with our \ours.
\begin{itemize}[leftmargin=*,]
    \item $w(\perm) \iff r(\tau)$: In our learning objectives, $w(\perm)$ plays a role similar to the reward $r(\tau)$ in RL. It represents the importance of the document permutation $\perm$ and serves as a weighting factor to evaluate the utility of documents to downstream tasks, analogous to how rewards $r(\tau)$ shape the policy model in RL.

    \item $p(\perm \mid x; \theta_{s}^{t}) \iff p(\tau \mid x; \pi_{\theta})$: The sampling distribution over permutations $p(\perm \mid x; \theta_{s}^{t})$ reflects the state of the selection model at iteration $t$, similar to how $p(\tau \mid x; \pi_{\theta})$ represents the current policy distribution over trajectories in RL. 

    \item $\log p(y, \perm \mid x; \theta) \iff \log p(\tau \mid x;\pi)$: The $\log p(y, \perm \mid x; \theta)$ in our Eq.~\ref{eq:update} is analogous to the $\log p(\tau \mid x;\pi)$ in RL. In both cases, the gradient is updated to increase the likelihood of actions (or permutations $\perm$) that yield higher rewards (or weights $w(\perm)$).
\end{itemize}

% \noindent%
From the RL perspective, the objective functions in Eq.~\eqref{eq:two-part} can be interpreted as a two-agent system collaborating within the RAG task.
In the estimation step, directly sampling a document permutation from the selection model $\theta_s$ is essentially a Monte Carlo process, while $w(\perm) = p(y \mid x, \doc_{\perm}; \theta_g)$ in the maximization step serves as the reward.
Similar to the RL, where policy models improve by reinforcing high-reward actions, the selection model $\theta_s$ improves by selecting documents that lead to better generation. The generator $\theta_g$, in turn, learns to leverage the selected documents to generate correct answers, creating a synergistic loop that optimizes the entire RAG performance.

\subsection{Impact of the importance sampling}\label{sec:is}
In the permutation estimation step, we apply an importance sampling strategy to directly sample document permutations from the knowledge selection model (i.e., $\perm \sim p(\perm \mid x;\theta_s)$) instead of the posterior distribution $p(\perm \mid x, y ; \theta)$.
We analyze its impact on the expectation and variance of the training objective.
For simplicity, we denote $\log p(y, \perm \mid x; \theta^t)$ in the $\elbo$ as a function $f(y,\perm)$.

\header{Expectation is unchanged}
In the permutation estimation step, we employ importance sampling to approximate the posterior distribution $p(\perm \mid x, y; \theta)$ with samples drawn from the proposal distribution $p(\mathbf{z} \mid \mathbf{x}; \theta_s)$, which is parameterized by the selection model.
To compensate for the discrepancy between the target and proposal distributions, we apply an importance weight the weight $w(\perm) = p(y \mid \doc_{\perm}, x; \theta)$ adjusts for the different between the target distribution $p(\perm \mid x; \theta_s^t)$ and the sampling distribution $p(\perm \mid x; \theta_s)$, which can be presented as:
\begin{equation}
    \begin{aligned}
        \mathbb{E}_{p(\perm \mid x, y; \theta^t)} [f(\perm)] &= \mathbb{E}_{p(\perm \mid x; \theta_s^t)} \left[ \frac{p(y \mid x,\doc_{\perm}; \theta^t)}{p(\perm \mid x; \theta_{s}^t)} f(\perm) \right] \\
        &= \mathbb{E}_{p(\perm \mid x; \theta_s^t)} \left[ w(\perm) f(\perm) \right].
    \end{aligned}
\end{equation}
The key property of importance sampling is that it provides an unchanged estimator of the expectation under the true posterior.
This fundamental property of importance sampling ensures that the expectation in the variational optimization objective remains unchanged.
 
\header{Variance decreases}\label{sec:variance}
While importance sampling preserves the expectation, it crucially affects the \textit{variance} of the estimator, which directly impacts the stability and efficiency of gradient-based training. We begin by formulating the \textit{vanilla variance}, i.e., the variance before applying importance sampling.

Before using importance sampling, the \textit{vanilla variance} $\text{Var}_{p(\perm \mid x, y; \theta^t)}\left[ f(y, \perm) \right]$ is formulated as:
\begin{equation}\label{eq:var-old}
\mathbb{E}_{p(\perm \mid x, y; \theta^t)} \left[ f(y,\perm)^2 \right] - \left( \mathbb{E}_{p(\perm \mid x; \theta_s^t)} \left[ f(y,\perm) \right] \right)^2.
\end{equation}
After using importance sampling, we denote the \textit{new variance} as $\text{Var}_{p(\perm \mid x; \theta_{s}^t)}\left[ w(\perm) f(y, \perm) \right]$, which becomes:
\begin{equation}\label{eq:var-new}
\mathbb{E}_{p(\perm \mid x; \theta_s^t)} \left[ w(\perm)^2 f(y,\perm)^2 \right] - \left( \mathbb{E}_{p(\perm \mid x; \theta_s^t)} \left[ w(\perm) f(y,\perm)\right] \right)^2.
\end{equation}
To evaluate the effect of importance sampling on variance, we derive the difference $\Delta \text{Var}$ between the new and old variances. 
\begin{equation}
\begin{aligned}
\Delta\text{Var} &= \text{Var}_{p(\perm \mid x; \theta_{s}^t)}\left[ w(\perm) f(y, \perm) \right] - \text{Var}_{p(\perm \mid x, y; \theta^t)}\left[ f(y, \perm) \right] \\
& = \mathbb{E}_{p(\perm \mid x,y;\theta^{t})} \left[f(y, \perm)^2 \left(  w(\perm) -1\right) \right].
\end{aligned}
\label{eq:var-comparison}
\end{equation}
Since $w(\perm) \propto  p(y \mid x, \doc_{\perm}; \theta_g^{t})$ and the $0 \leq p(y \mid x, \doc_{\perm}; \theta_g^{t}) \leq 1$, it follows that $\Delta\text{Var} \leq 0$ if  $w(\perm)$ is normalized.
This expression reveals a key insight: if $w(\perm) \leq 1$ for all $\perm$, then $\Delta \text{Var} \leq 0$, indicating that the variance of the importance-weighted estimator is strictly \textit{lower} than the original variance.
% In practice, we normalize the $w(\perm)$ to decrease variance after importance sampling.
o further guarantee numerical stability and prevent rare outlier samples with disproportionately large weights, we apply normalization to $w(\perm)$ in practice:
\begin{equation}
    w(\perm) = \frac{w(\perm)}{\sum_{\perm'} w(\perm')}.
\end{equation}
This normalization step ensures bounded gradients and facilitates smoother convergence throughout the DRO training process.

\header{Training process gradually stabilizes}\label{sec:stability}
From Eq.~\eqref{eq:var-old} and Eq.~\eqref{eq:var-new}, we notice that the $w(\perm)$ plays a key role in shaping the variance, which affects the training stability.
Since $w(\perm) \propto  p(y \mid x, \doc_{\perm}; \theta_g^{t})$, we have the following analysis to demonstrate the training stability of our \ours:
\begin{enumerate*}[label=(\roman*)]
    \item initially, the generator $\theta_g^t$ is less optimized, which potentially results in high variance in $p(y \mid x, \perm; \theta_g^t)$; 
    \item as generator $\theta_g^t$ is optimized through the loss function in Eq.~\eqref{eq:gen-opt}, $p(y \mid x, \perm; \theta_g^t)$ stabilizes, which leads to a progressively reduced variance.
\end{enumerate*}
% We also provide a quantitative experiment in Section~\ref{sec:stability}.

\subsection{Convergence of the optimization process}\label{sec:convergence}

To prove the convergence of \ours, we show the non-decreasing and upper-bounded property of  $\log p(y | x, \theta^{t})$ during the training. 

We first prove that the log-likelihood $\log p(y \mid x, \theta^{t+1})$ is non-decreasing after each training iteration, which is formulated as: $\log p(y\mid x, \theta^{t+1}) \geq \log p(y | x, \theta^{t})$. Here $\theta = \left(\theta_{s}, \theta_{g}\right)$ consists of the parameters of the selection model $\theta_s$ and the LLM generator $\theta_g$.

\begin{proof}
At each iteration $t$, we start by estimating the distribution of the latent variable $\perm$, i.e., document permutation.
Since we apply an importance sampling strategy to directly sample $\perm$ from the selection model $\theta_{s}^{t}$, the $\elbo$ is transformed as in Eq.~\eqref{eq:elbo-is}:
\begin{equation*}
    \begin{aligned}
        \elbo(\theta, \theta^{t}) &:= \mathbb{E}_{\perm \sim p(\perm \mid x;\theta_{s}^{t})}\left[w(\perm) \log p(y, \perm \mid x ; \theta)\right].
    \end{aligned}
\end{equation*}
In the maximization step, we update $\theta$ by maximizing the ELBO: $\theta^{t+1} = \arg \max_{\theta} \elbo(\theta, \theta^{t})$. This ensures that $ \elbo(\theta^{t+1}, \theta^{t}) \geq \elbo(\theta^{t}, \theta^{t})$.
We further observe that the marginal log-likelihood in Eq.~\ref{eq:decomp} satisfies:
\begin{equation}
    \begin{aligned}
        \log p(y \mid x; \theta^{t+1}) & \geq \elbo(\theta^{t+1}, \theta^{t}) \\
       & \geq \elbo(\theta^{t}, \theta^{t}) = \log p(y \mid x; \theta^{t}).
    \end{aligned}
\end{equation}
Thus, we establish that $\log p(y \mid x; \theta^{t+1}) \geq \log p(y \mid x; \theta^{t})$, demonstrating the non-decreasing nature of the optimization process.
Next, we examine the boundedness of the log-likelihood. Given that $0 \leq p(y \mid x; \theta) \leq 1$, it follows that $-\infty \leq \log p(y \mid x; \theta) \leq 0$.
Then, we introduce an existing theorem from~\cite{Bibby_1974} as follows.
\begin{theorem}
\label{thm:converge}
\textbf{Monotone Convergence Theorem:} If a sequence $\{a_n\}$ is monotonic (either non-decreasing or non-increasing) and bounded, then it converges to a finite limit.
\end{theorem}

\noindent%
Applying Theorem~\ref{thm:converge}, the non-decreasing and upper-bounded nature of $\{\log p(y \mid x; \theta^t)\}_{t=1}^\infty$ ensures that the sequence converges to a finite limit, proving the convergence of the \ours training.
In practice, when the performance on the validation set shows no further improvement, the model is considered to have converged.
\end{proof}

\subsection{Upshot}

Below, we summarize three key insights from the above theoretical analysis:

\begin{enumerate}[label=(\roman*)]
    \item \todo{The objective of \ours can be interpreted as optimizing a collaborative two-agent system within the RAG framework. The selection model $\theta_s$ improves by identifying document permutations that enhance the generator's performance. In turn, the generator $\theta_g$ learns to better utilize the selected documents to produce correct answers. This mutual improvement forms a synergistic loop that jointly optimizes the overall RAG performance.}

    \item \todo{The use of importance sampling in \ours preserves the expectation of the learning objective while reducing variance, particularly when the importance weights are normalized. This contributes to more stable and efficient training.}

    \item \todo{The convergence of \ours is theoretically guaranteed by the monotonic increase of the log-likelihood and its upper boundedness. In practice, this implies that the model either improves or maintains its performance over iterations, and training can be safely terminated when validation performance saturates, ensuring both convergence and training efficiency.}
\end{enumerate}

\section{Experimental setup}

\subsection{Datasets}

Following previous work~\citep{Lewis2020RetrievalAugmentedGF, li2024corpuslm, zamani2024stochastic, yu2024rankrag}, we conduct experiments on \textit{full development set} of five commonly used question-answering benchmarks: Natural Question (NQ)~\cite{kwiatkowski-etal-2019-natural},
HotpotQA~\citep{yang2018hotpotqa},  MuSiQue~\citep{trivedi2021musique}, 2Wikimultihopqa (2WikiQA)~\citep{xanh2020_2wikimultihop}, and 
Wizard-of-Wikipedia (WoW) \citep{dinan2018wizard}.
Table~\ref{tab:dataset} presents the statistics of these datasets.

\subsection{Evaluation metrics}
In line with previous studies~\citep{li2024rag, wei2024instructrag, shi2024generate}, we use \textit{F1}  and \textit{Exactly Match} (EM) metrics from KILT~\cite{petroni-etal-2021-kilt} for evaluation.
The \textit{F1} score is used to measure the token-level 
overlap between the generated answer and the ground truth answer, which represents the harmonic mean of precision and recall, where the recall is determined by considering the number of overlaps with the correct answer tokens, while precision is determined by considering the number of overlaps with all generated tokens.
The \textit{EM} metric checks if the predicted strings exactly match the ground truth.
Besides the end-to-ene evaluation, we also use the \textit{Recall@K} (\textit{K}=1,3,5) to evaluate the document re-ranking performance of our selection model following previous work~\cite{erqrag, petroni-etal-2021-kilt, li2024corpuslm}, in which the Recall@K is set to 1 if and only if the top-K documents contains the ground-truth answer.

\subsection{Baselines}\label{sec:baseline}
We compare the proposed \ours with four categories of baselines \todo{based on how they integrate the knowledge selection process with the answer generation process.}

\header{Prompting-based methods}
These methods guide LLMs to leverage external knowledge through prompt learning.
We evaluate:
\begin{enumerate}[label=(\roman*)]
    \item \textit{RetGen}~\cite{shao2023enhancing}, which interleaves query formulation, retrieval, and answer generation iteratively;
    \item \textit{GenGround}~\cite{shi2024generate}, which instructs a LLM to generate answers and then revise them using retrieved documents;
    \item \textit{In-context RAG}~\cite{icralm}, which truncates the top k retrieved documents as context for the LLM to generate answers.
\end{enumerate}
For the first two baselines, we use GPT-3.5-turbo as the backbone, following their official implementations.
For \textit{in-context RAG}, we evaluate various LLMs with the same three in-context examples.

\header{Retrieval Tuning}
These methods improve the entire RAG performance by enhancing the retrieval process. 
We evaluate: 
\begin{enumerate}[label=(\roman*)]
    \item \textit{REPLUG}~\cite{shi2023replug} and \textit{DPA-RAG}~\cite{dong2024understand}, which adapt a dense retriever or point-wise ranker to a frozen LLMs by re-scoring the relevance;
    \item \textit{FLICO}~\cite{wang2023learning} and \textit{RECOMP}, which filter irrelevant content from retrieved documents; and
    \item \textit{Re-ranking Models}, which re-rank retrieved documents, passing the top-\textit{K} ranked results for LLMs.
\end{enumerate}
We benchmark point-wise (MonoT5~\cite{nogueira2020document} and BGE~\cite{bge_embedding}) and list-wise rankers ( RankVicuna~\cite{pradeep2023rankvicuna} and RankZephyr~\cite{pradeep2023rankzephyr}).

\header{LLM Fine-tuning}
These methods fine-tune LLMs through instruction tuning, improving LLMs' ability to utilize retrieved documents.
We evaluate: 
\begin{enumerate}[label=(\roman*)]
    \item Vanilla supervised fine-tuning (SFT),  which train LLMs by maximizing the answer likelihood based on query and documents;
    \item  \textit{Chat\-QA}~\cite{liu2024chatqa} and \textit{RankRAG}~\cite{yu2024rankrag}, which post-train LLMs on diverse knowledge-grounded tasks;
    \item  \textit{RetRobust}~\cite{yoran2023making}, \textit{InstructRAG}~\cite{wei2024instructrag}, \textit{Self-RAG}~\cite{asai2023self} and \textit{RAAT-7B}~\cite{fang2024enhancing}, which train LLMs to identify relevant content from retrieved documents for QA.
\end{enumerate}

\header{End-to-end training}
These methods train RAG components with an end-to-end objective.
We benchmark:
\begin{enumerate}[label=(\roman*)]
    \item \textit{Atlas}~\cite{izacard2023atlas}, a pre-trained retrieval-augmented LLM;
    \item  \textit{RA-DIT}~\cite{lin2023ra}, which initially trains a generator and subsequently fine-tune a dual-encoder retriever; and
    \item \textit{DDR-RAG}~\cite{li2024rag}, which employs DPO~\cite{rafailov2024direct} (Direct Preference Optimization) to jointly train a point-wise ranker and an answer generator.
\end{enumerate}

To ensure fairness,  we set the size of retrieval documents as 20 for all the baselines and the $K=5$, aligning with the implementation of \ours.
Since the code or model checkpoints are unavailable for some baselines, we mark their incomplete results as ``--''.
In such cases, we report results from the original papers but only for reference.

\begin{table}[!t]
\caption{Statistics of our experimental datasets, where \todo{we provide the amount of training and evaluation dataset, the average length of input query (word) as well as the retrieval corpus.}}
\label{tab:dataset}
\centering
\begin{adjustbox}{width=\columnwidth,center}
\setlength\tabcolsep{10pt}
\begin{tabular}{l cc cc c}
\toprule
\begin{tabular}[c]{@{}c@{}} \textbf{Experimental} \\ \textbf{Benchmarks} \end{tabular}
& \begin{tabular}[c]{@{}c@{}} \textbf{Training} \\ \textbf{Data Size} \end{tabular} 
& \begin{tabular}[c]{@{}c@{}} \textbf{Query Length} \\ \textbf{(Train)} \end{tabular} 
& \begin{tabular}[c]{@{}c@{}} \textbf{Evaluation} \\ \textbf{Data Size} \end{tabular} 
& \begin{tabular}[c]{@{}c@{}} \textbf{Query Length} \\ \textbf{(Evaluation)} \end{tabular} 
& \begin{tabular}[c]{@{}c@{}} \textbf{Retrieval} \\ \textbf{Corpus} \end{tabular}
\\
\midrule
Nature Question~\cite{kwiatkowski-etal-2019-natural}
& \phantom{0}58,622 & \phantom{0}9.21
& \phantom{0}6,489  & \phantom{0}9.16
& Wiki2018 
\\

Hotpot QA~\cite{yang2018hotpotqa}
& \phantom{0}90,185 & 17.85
& \phantom{0}7,384  & 15.63 
& Wiki2018
\\

MusiQue QA~\cite{trivedi2021musique} 
& \phantom{0}19,938 & 15.96
& \phantom{0}2,417  & 18.11
& Wiki2018
\\
 
2WikiMultiHopQA~\cite{xanh2020_2wikimultihop} 
& 167,454 & 12.74
& 12,576  & 11.97
& Wiki2018
\\

Wizard-of-Wikipedia~\cite{dinan2018wizard}
& \phantom{0}63,734  & 70.41
&  \phantom{0}3,054  & 70.25
& Wiki2018
\\
\bottomrule
\end{tabular}
\end{adjustbox}
\end{table}

\subsection{Implementation details}\label{sec:implement}
We use the Llama-3-8B~\cite{dubey2024llama} as the backbone model for both our ranker and LM generator.
We also alternate it with another model, i.e., Mistral-7B~\cite{mistral}, to evaluate the generalizability of \ours across different backbones.
Given a query $x$, we initially use the off-the-shelf ColBERTv2.0~\cite{colbertv2} to retrieve top $20$ documents as input for the selection model $\theta_s$, which  selects a permutation containing $K=5$ documents to the generator $\theta_g$.
Following previous work~\cite{xu2023searchinthechain, li2024rag, karpukhin2020dense}, the document corpus for the retrieval is based on the Wikipedia passage dump from Dec. 20, 2018.
We set the maximum training iteration $N=5$ (\S~\ref{sec:maximization}) and report the performance for each iteration checkpoint for a comprehensive analysis.
The sampling number for document permutation is set to 8.
We use the DeepSpeed ZeRO strategy~\cite{deepspeed} during training, with learning rates of $1e^{-5}$ and a weight decay coefficient of 0.01.
The training of the \ours can be done within 20 hours with 8 NVIDIA A100-PCIE-80GB GPUs.

\begin{table*}[!t]
\centering
\small
\caption{Experimental results on five benchmarks, where we highlight the best results in \underline{bold}. We also compute the average EM or F1 across all datasets. \textit{Scale} indicates the parameter size of the LLM generator. $^*$ denotes the baselines based on closed-source \textit{gpt-3.5}.
$^\dagger$ and $^\ddagger$ indicate significant improvements over best open-source baselines with p-value < 0.05 and 0.01.}
\label{tab:main}

\setlength\tabcolsep{2.0pt}
\begin{adjustbox}{width=\columnwidth,center}
\begin{tabular}{@{}p{3.7cm} r cc cc cc cc cc cc@{}}
\toprule
\textbf{\textcolor{blue}{Tasks}}&
& \multicolumn{2}{c}{\textbf{NQ}} 
& \multicolumn{2}{c}{\textbf{HotpotQA}} 
& \multicolumn{2}{c}{\textbf{MuSiQue}} 
& \multicolumn{2}{c}{\textbf{2WikiQA}} 
& \multicolumn{1}{c}{\textbf{WoW}}
& \multicolumn{2}{c}{\textbf{Avg.}}\\
\midrule
\textbf{Metrics}  & \textbf{Scale} 
& \textbf{EM} & \textbf{F1}
& \textbf{EM} & \textbf{F1}
& \textbf{EM} & \textbf{F1}
& \textbf{EM} & \textbf{F1}
& \textbf{F1}
& \textbf{EM} & \textbf{F1} \\

\hline
\rowcolor{Gainsboro}
\multicolumn{13}{l}{\textit{Prompting-based Method}} \\

RetGen$^*$~\citep{shao2023enhancing} & 175B
 & 37.75  & 39.78 & 39.04 &  42.10  &17.69  &21.04   &  33.00   & 39.17 & 16.97
& 31.87  & 31.81
\\

GenGround$^*$~\citep{shi2024generate} & 175B 
& 40.60 & 42.31   & 41.27  & 44.71 & 20.77 & 24.36  &  39.61 & 42.58 & 17.76
& 35.56 & 34.34
\\

In-context RAG~\cite{icralm} \\
\textit{ - w/} Llama3-inst.-70B~\cite{dubey2024llama} & 70B
 &  39.38 & 47.26 &  37.38  &  39.62 & 16.43  & 21.16 &  37.26 & 41.46 & 17.06
 &  32.61 & 33.31
\\ 

\textit{ - w/} Llama2-70B-chat~\cite{touvron2023llama}& 70B 
& 38.07 & 44.69 & 37.14 &40.27 & 16.78 & 20.11 & 38.51 &  41.02 &15.75 
& 32.62 & 32.37
\\ 

\textit{ - w/} Mixtral-inst.-8x7B~\cite{jiang2024mixtral} & 56B 
& 39.34 & 46.34 & 39.63 & 42.53 & 16.65 & 20.73&  37.03 & 38.50 & 16.66
& 33.16 & 32.95
\\

\textit{ - w/} Llama2-13B-chat~\cite{touvron2023llama} &  13B 
& 35.27 & 41.54 & 35.43 & 40.37 & 16.87 & 18.37 & 35.47 &    38.63 &   13.12 
& 30.76 & 30.41
\\

\hline
\rowcolor{Gainsboro}
\multicolumn{13}{l}{\textit{Retrieval tuning}} \\

REPLUG~\cite{xu2023recomp}
& 65B & 28.80 &  -- & 32.00 &  --& -- &  --& -- &  -- & --
& -- & --
\\

% ---
RECOMP~\cite{xu2023recomp} & 20B
&  37.47 & 42.67 & 38.72 & 42.72 & 17.34 &  24.96 &  32.17   &  38.26 & 15.11
& 31.43 & 32.74
\\

DPA-RAG~\cite{dong2024understand} & 8B 
& 37.29   & 44.31 &  37.15 &  40.53 & 18.45 & 20.36  & 39.02 & 39.66 & 14.73
& 32.98 & 31.92
\\

FLICO~\cite{wang2023learning}& 7B
& 35.32  & 37.24 & 38.74 & 39.21 & 14.53 &  16.81 & 29.74 & 33.05 & 14.23
& 29.58 & 28.11
\\

MonoT5~\citep{nogueira2020document} \textit{w/} Mistral$\dagger$ & 7B
&  31.78 &  37.68 &  32.38   & 38.76 & 14.31 &   19.16 & 35.96 & 37.11 & 14.27
& 28.61 & 29.40
\\

RankVicuna~\cite{pradeep2023rankvicuna} \textit{w/} Mistral$\dagger$ &  7B 
& 33.78 & 39.58& 34.74 & 41.25 &15.38 &  21.87 & 36.72 &  38.93 & 14.45
& 30.16 & 31.22
\\

RankZephyr~\cite{pradeep2023rankzephyr} \textit{w/} Mistral$\dagger$ &  7B
& 34.77 &  45.22 & 36.08 & 42.77 & 16.03  &   21.75 & 35.07  & 38.31  & 14.94
& 30.49 & 32.60
\\

BGE-ranker~\citep{bge_embedding} \textit{w/} Mistral$\dagger$ &  7B
& 35.21 & 40.50 & 37.61 & 38.10 & 16.61 &  21.37 & 37.16 & 38.49 & 14.77
& 31.65 & 30.65
\\

\hline
\rowcolor{Gainsboro}
\multicolumn{13}{l}{\textit{LLM Fine-tuning}} \\

RetRobust~\cite{yoran2023making} & 13B 
&  37.03 & 43.82 & 35.59  & 40.54  & 18.11  &  18.16   & 38.65  &  39.11  & 17.04
& 32.34 & 31.73
\\

ChatQA~\cite{xu2023recomp} &  8B
 & 23.64 & 34.54& 33.40 & 44.60 & 16.64 & 17.05  & 26.80 & 31.90 & 16.22
&25.12 & 28.86
\\

RankRAG~\cite{xu2023recomp} & 8B 
& -- & -- & 35.30  & 46.70 & -- &  -- & 31.40 & 36.90 & -- 
& -- & --
\\

InstructRAG~\cite{wei2024instructrag}  & 8B 
&  35.90  &   38.21  &  30.69 & 34.71 & 14.94 &  25.88 & 35.92  & 20.01  & 14.57
&  29.36 & 26.68
\\

Vanilla SFT \textit{w/} Llama3~\cite{touvron2023llama}  & 8B
& 35.25 & 38.46 & 25.07 & 32.57 & 14.35 & 17.82  & 30.65 & 30.43  & 13.91
& 26.33 & 26.64
\\

Vanilla SFT \textit{w/} Mistral~\cite{mistral} & 7B
& 34.65 &37.52 & 25.75 & 30.65 & 13.35 & 17.97  & 30.43  &  30.65   & 13.83
& 26.05 & 26.12
\\

Self-RAG ~\cite{asai2023self} &  7B 
& 29.74 & 31.63   & 16.30 & 27.30 & 9.43  & 21.50 &23.52  & 27.33&13.24
&19.75 & 24.2
\\

RAAT~\cite{fang2024enhancing}& 7B
& 33.53  & 37.85 &  33.01 & 31.67 & 17.08 &  21.793 &  29.69 & 32.68  & 15.37
& 28.33 & 27.87
\\

\hline
\rowcolor{Gainsboro}
\multicolumn{13}{l}{\textit{End-to-end optimization}} \\

RA-DIT~\cite{lin2023ra}
&  65B & 35.20 & -- & 39.70  & --  & -- &  -- & -- & -- & -- 
& --  & --
\\

Atlas~\cite{izacard2023atlas}
& 11B & 26.70 & -- & 34.70 & -- & -- & -- & -- & --& --
& -- & --
\\

DDR-RAG~\cite{li2024rag} &  8B 
&  40.74  & 28.76 &  31.71 & 40.04 & 13.54 & 10.57 &35.44 & 38.40 & 16.21
& 30.36 & 26.80
\\

\rowcolor{green!12}\ours-Mistral-7B (\textbf{Ours}) & 7B 
& 42.41 & 51.01 & 40.37 & 47.87  & \underline{\textbf{21.36}} & 25.32 & \underline{\textbf{42.12}}  & 43.65& 18.31
& 36.56 & 37.23
\\

\rowcolor{green!12}\ours-Llama-3-8B (\textbf{Ours}) & 8B
& \textbf{\underline{45.76\rlap{$^\ddagger$}}}  &	\textbf{\underline{55.42\rlap{$^\ddagger$}}} & \textbf{\underline{42.23\rlap{$^\dagger$}}} & \textbf{\underline{49.27\rlap{$^\ddagger$}}}  &  20.64\rlap{$^\dagger$} &  \textbf{\underline{25.97\rlap{$^\dagger$}}} & 40.12\rlap{$^\dagger$} & \textbf{\underline{44.12\rlap{$^\ddagger$}}} & \textbf{\underline{18.76\rlap{$^\ddagger$}}}
& \textbf{\underline{37.19\rlap{$^\ddagger$}}} & \textbf{\underline{38.71\rlap{$^\ddagger$}}}
\\

\bottomrule
\end{tabular}
\end{adjustbox}
\end{table*}

\section{Experimental results}\label{sec:experiment-results}

\subsection{Overall performance}

\header{Single-hop QA benchmark}
We first analyze the performance of the proposed \ours on open-domain QA and dialogue datasets, including NQ and WoW.
The inputs in these datasets are complex, posing challenges for neural models to comprehend~\cite{kwiatkowski-etal-2019-natural}.
Table~\ref{tab:main} presents the results, where we observe that our \ours achieves the best performance, such as pushing the EM from 40.74 to 45.76 in NQ dataset (12.32\% relative improvements).
To investigate the reasons for our superior performance, we conduct a coarse-to-fine analysis.
We begin by identifying some evaluation cases that incorrectly answered by strong baselines (e.g., InstructRAG) while correctly answered by our method.
Then we compare the outputs from both approaches for these cases.
The selection model, trained collaboratively with the LLM generator in \ours, selects more useful documents for answer generation.
This demonstrates that our selection model learns to understand utility-oriented document selection criteria in RAG tasks, contributing to holistic improvement.

\begin{figure}[!t]
    \centering
	\includegraphics[width=0.95\linewidth]{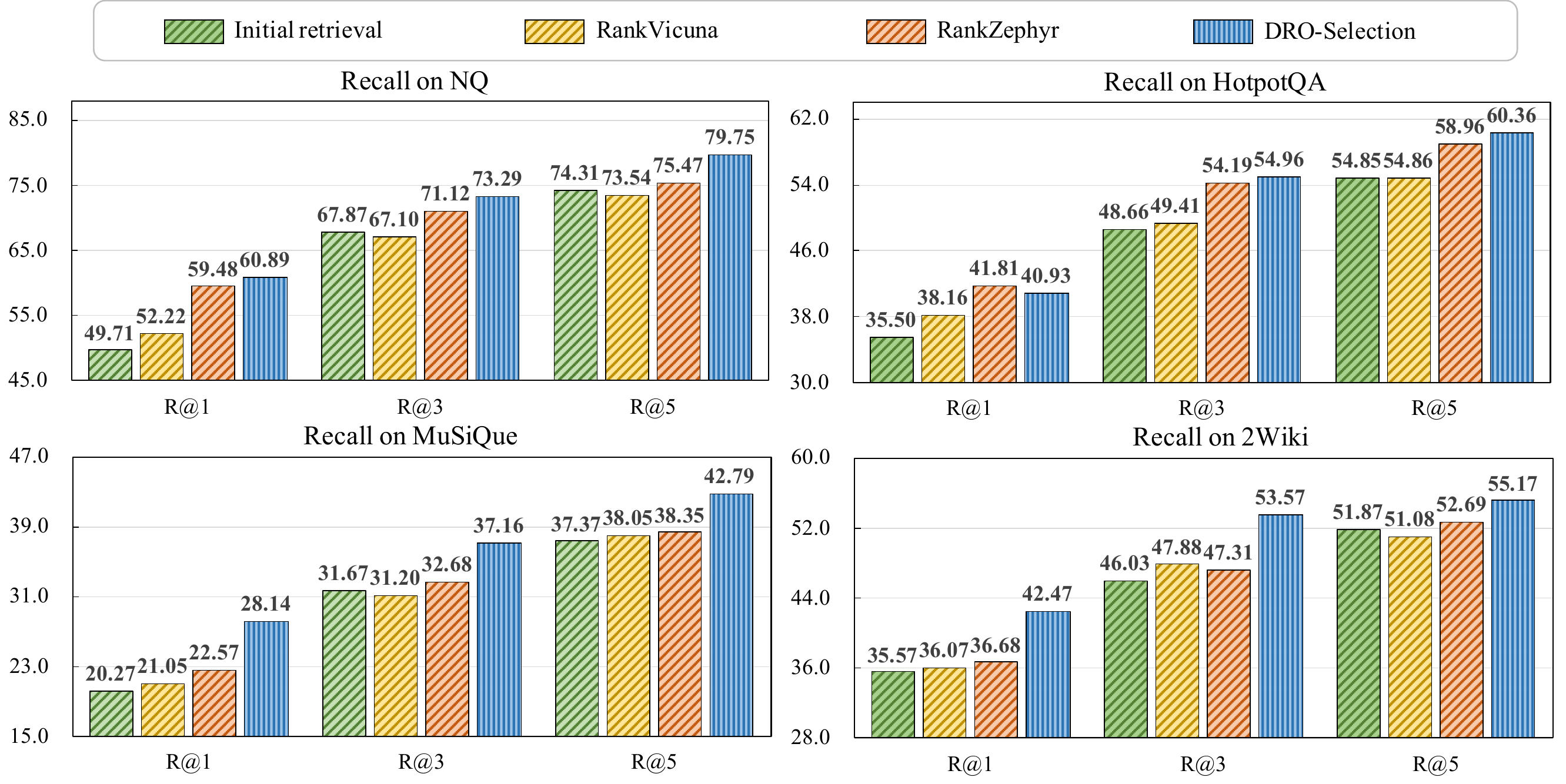}
    \caption{Recall@\textit{K} (k=1, 3, 5) score of the initial retrieval (Colbertv2.0), two re-ranking baselines (i.e., RankVicuna and RankZephyer) and our selection model $\theta_{s}$, respectively.}
    \label{fig:recall}
\end{figure}

\header{Multi-hop QA benchmark}
We further validate the superiority of the proposed method on multi-hop QA tasks, which are more challenging since they require the model to aggregate evidence from multiple documents to derive the answer.
For example, based on our statistic, queries in HotpotQA involve at least 2.21 hops, while those in MuSiQue require an average of 2.65 hops.
As shown in Table~\ref{tab:main}, our \ours substantially outperforms existing baselines, such as achieving a F1 score of 49.27 and an EM score of 42.23  in HotpotQA dataset.
The reasons for this improvement are: 
(1) the weight $w(\perm)$ in the selection model's learning objective severs as a reward function, reinforcing selecting documents that maximize end-to-end generation performance, and 
(2) the LLM generator is optimized synchronously, further enhancing overall effectiveness.

\header{Document selection evaluation}
In addition to the end-to-end evaluation presented in Table~\ref{tab:main}, we further evaluate the Recall@\textit{K} of our selection models on documents re-ranking task.
Following~\cite{karpukhin2020dense, kwiatkowski-etal-2019-natural}, we set the Recall@K to 1 if the top-K documents contains the ground-truth answer.
As shown in Figure~\ref{fig:recall}, our selection model improves the recall of the initial retrieval of ColBERTv2.0 by on average 17.78\%, such as pushing the accuracy@1 from 49.71 to 60.89 on NQ dataset.
We also compared with the similar list-wise re-ranking baselines, where we find that the selection model, trained within \ours, achieves the highest recall.
This validates the need to enable end-to-end supervision for knowled    ge selection process.

\begin{table}[htbp]
\centering
\caption{F1 score for checkpoints in each training iteration.}
\begin{adjustbox}{width=0.85\columnwidth,center}
\setlength\tabcolsep{14pt}
\begin{tabular}{@{}c cccc c@{}}
\toprule
\textbf{Iteration} & \textbf{NQ} & \textbf{HotpotQA} & \textbf{MuSiQue} &\textbf{2WikiQA}  & \textbf{Average $\Delta$ }
\\ 

\hline
\rowcolor{Gainsboro}\multicolumn{6}{l}{\textit{\ours-Mistral-7B}} \\
1 &  41.56  &  39.65  &  20.79  &  36.5 & -- \\
2  &  47.19$_{\uparrow 11.9\%}$  &  43.71$_{\uparrow 9.3\%}$  &  23.26$_{\uparrow 10.6\%}$  &  39.53$_{\uparrow 7.7\%}$ & 9.8\% \\
3  &  49.98$_{\uparrow 5.6\%}$  &  46.03$_{\uparrow 5.0\%}$  &  24.87$_{\uparrow 6.9\%}$  &  41.54$_{\uparrow 4.8\%}$  & 5.5\% \\
4  &  50.82$_{\uparrow 1.7\%}$  &  47.24$_{\uparrow 2.6\%}$  &  25.07$_{\uparrow 0.8\%}$  &  41.75$_{\uparrow 0.5\%}$ & 1.4\% \\
5  &  50.97$_{\uparrow 0.3\%}$  &  47.87$_{\uparrow 1.3\%}$  &  25.32$_{\uparrow 1.0\%}$  &  42.12$_{\uparrow 0.9\%}$ & 0.9\%\\

\hline
\rowcolor{Gainsboro}\multicolumn{6}{l}{\textit{\ours-Llama-3-8B}} \\
1  &  44.84  &  41.07  &  21.43  &  38.43 & -- \\
2  &  49.11$_{\uparrow 8.7\%}$  &  44.73$_{\uparrow 8.2\%}$  &  24.06$_{\uparrow 10.9\%}$  &  41.27$_{\uparrow 6.9\%}$ & 8.7\% \\
3  &  54.98$_{\uparrow 10.7\%}$  &  48.24$_{\uparrow 7.3\%}$  &  25.24$_{\uparrow 4.7\%}$  &  43.02$_{\uparrow 4.1\%}$ &  6.7\% \\
4  &  55.01$_{\uparrow 0.1\%}$  &  49.06$_{\uparrow 1.7\%}$  &  25.94$_{\uparrow 2.7\%}$  &  43.65$_{\uparrow 1.4\%}$ &  1.5\% \\
5  &  55.42$_{\uparrow 0.7\%}$  &  49.27$_{\uparrow 0.4\%}$  &  25.97$_{\uparrow 0.1\%}$  &  44.12$_{\uparrow 1.1\%}$ &  0.6\% \\
\bottomrule
\end{tabular}
\label{tab:converge}
\end{adjustbox}
\end{table}

\subsection{Training Convergence}

% \paragraph{Convergence.}
As illustrated in Algorithm~\ref{algo:opt}, \ours iteratively optimizes the model parameters $\theta = (\theta_s, \theta_g)$ to improve performance. To analyze convergence, we evaluate the model checkpoint at each training iteration.
Table~\ref{tab:converge} presents the results on four experimental results.
For both \ours-Mistral-7B and \ours-LLaMA-3-8B, we observe a consistent and substantial increase in F1 scores during the first three iterations. 
From iteration 3 to 5, the gains begin to taper off, indicating the model has entered a convergence phase. Notably, \ours outperforms competitive baselines such as BGE and RankZephyr after just two iterations, demonstrating the promising performance of our method.

\subsection{Training Stability}
In our maximization step, we employ an importance weight $w(\perm)$ to calibrate the optimization expectation.
Our theoretical analysis in Section~\ref{sec:variance} shows that the training variance decreases as training progresses, since $w(\perm) \propto p(y \mid x, \doc_{\perm}; \theta_g)$.
To validate this finding, we compute the variance of sampling from $p(y \mid x; \doc_{\perm})$ at each training iteration.
We vary the number of samples from 1 to 16 to compute the variance, respectively.
See Figure~\ref{fig:variance}.
The variance substantially decreases with the optimization of model $\theta_g$ during training progress, validating the correctness of our theoretical analysis from Section~\ref{sec:stability}.
\todo{We also find that increasing the number of samples per iteration can reduces variance at each training step, which is an straightforward solution to improve the training stability.
Besides, in our experiments, we observe that increasing the number of samples per iteration reduces variance at each training step, offering a straightforward strategy to improve training robustness, especially during early-stage training (See Section~\ref{sec:sample-number} from more details).}

\begin{figure}[!t]
    \centering
	\includegraphics[width=\linewidth]{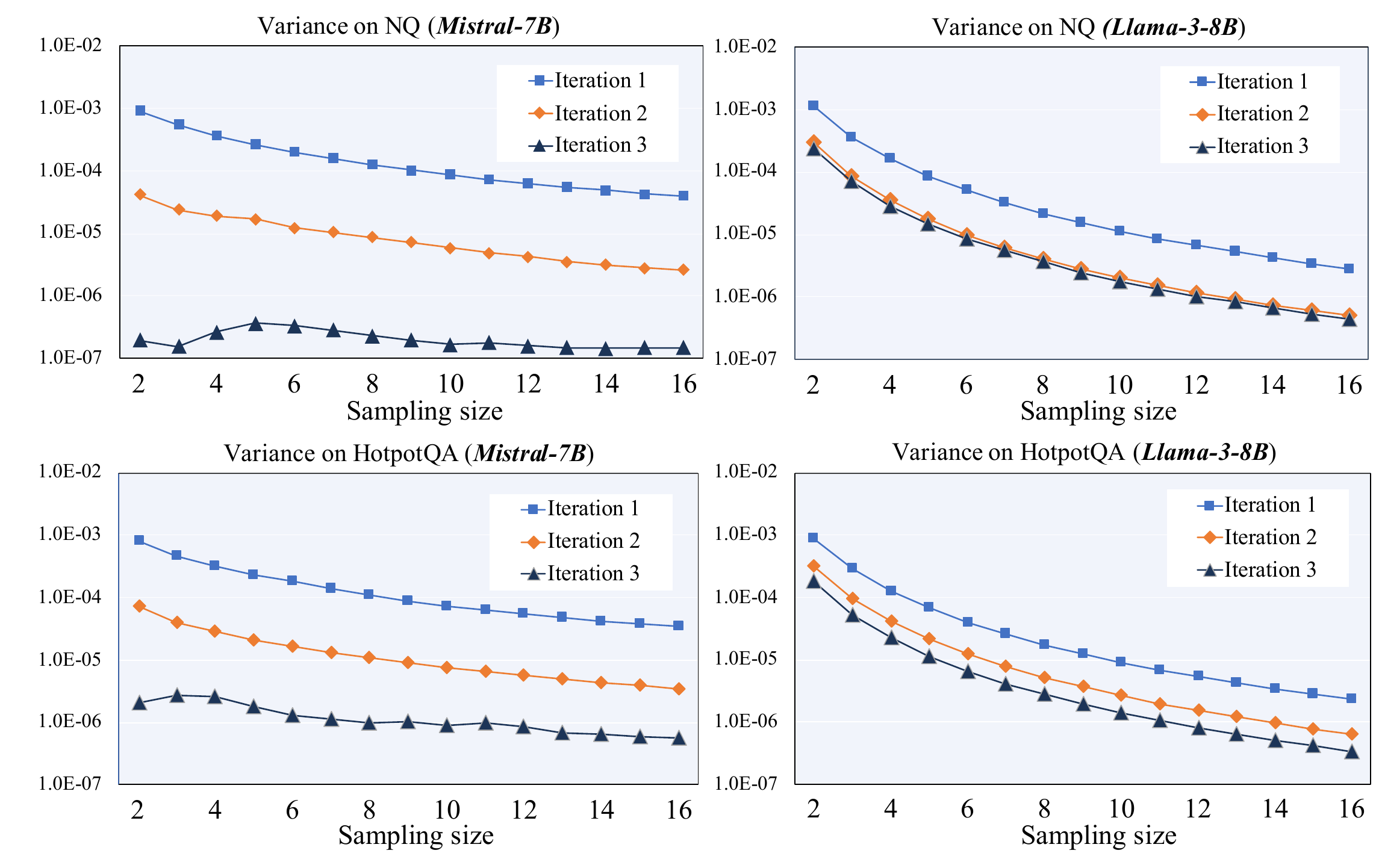}
    \caption{Variance during the training process of our method (logarithmic scale).}
    \label{fig:variance}
\end{figure}

\begin{remark}
\todo{The reduction of variance during training is a direct consequence of the improved confidence of the generator $\theta_g$ in predicting correct answers.
As $p(y \mid x, \doc_{\perm}; \theta_g)$ stabilizes, the importance weight $w(\perm)$ becomes less volatile, leading to lower variance and enhanced training stability. This validates the use of importance weighting in \ours as not only theoretically sound but also practically stabilizing.}
\end{remark}

\subsection{Ablation studies}
% The proposed \ours method enables the end-to-end optimization of two components: the selection model $\theta_s$ and the LLM generator $\theta_g$.  
% To analyze the effectiveness of each component, we conduct an ablation study by training only one component while keeping the other frozen.  
% Figure~\ref{fig:ablation} presents the experimental results. 
% We observe that training only the generator or the selection model leads to limited performance compared to our vanilla method.  
% For example, the variant \ours-\textit{w/o selector} only achieves F1${}=46.77$ (a decrease of 4.24), while the variant \ours-\textit{w/o generator} achieves F1${}=48.71$ (a decrease of 2.90) on the NQ dataset. These results further validate the advantage of \ours.
% \subsection{Ablation Studies}

The proposed \ours method enables end-to-end optimization of two tightly coupled components: the selection model $\theta_s$ and the LLM generator $\theta_g$. To quantify the individual contributions of these components, we conduct an ablation study by isolating the training of each module while freezing the other.
Formally, we compare the full \ours method against two ablated variants:
\begin{enumerate}[label=(\roman*)]
    \item \textbf{\ours-w/o Selector}: only updates the generator $\theta_g$, keeping $\theta_s$ fixed.
    \item  \textbf{\ours-w/o Generator}: only updates the selector $\theta_s$, keeping $\theta_g$ fixed.
\end{enumerate}

Figure~\ref{fig:ablation} reports the F1 scores across five datasets (NQ, HotpotQA, MuSiQue, 2WikiMultihopQA, and WoW) over five training iterations. We observe consistent and notable improvements of the full \ours method over both ablated variants.
On the \textbf{NQ} dataset, the vanilla method achieves an F1 score of \textbf{51.01} at iteration 5, while the two variants, i.e., \ours-\textit{w/o Selector} and \ours-\textit{w/o Generator}, obtain 46.77 and 48.71 respectively, suffering from drops of 4.24 and 2.30 points. 
Similar patterns are observed in other datasets:
\begin{enumerate}[label=(\roman*)]
    \item \textbf{HotpotQA}: \ours achieves 47.87 F1 vs. 44.60 (\ours-\textit{w/o Selector}) and 43.11 (\ours-\textit{w/o Generator}).
    \item \textbf{MuSiQue}: full model achieves 25.32, outperforming 22.59 and 21.77.
    \item \textbf{2Wiki}: the full pipeline reaches 43.65 F1, compared to 40.12 and 39.01.
    \item \textbf{WoW}: \ours improves F1 to 18.31, surpassing 16.47 and 15.73.
\end{enumerate}
These results demonstrate two key findings. First, training the selection model $\theta_s$ is critical for learning high-quality document permutations that benefit the generator. Second, optimizing the generator $\theta_g$ to better utilize selected documents further amplifies performance. The combined training of both components yields cumulative gains that neither alone can achieve.
\begin{remark}
    \todo{The ablation study highlights the necessity of joint optimization in \ours. It validates the design principle that the selector and generator must co-evolve during training to fully realize the benefits of retrieval-augmented generation.}
\end{remark}

\begin{figure*}[!t]
        \centering
	\includegraphics[width=
 \linewidth]{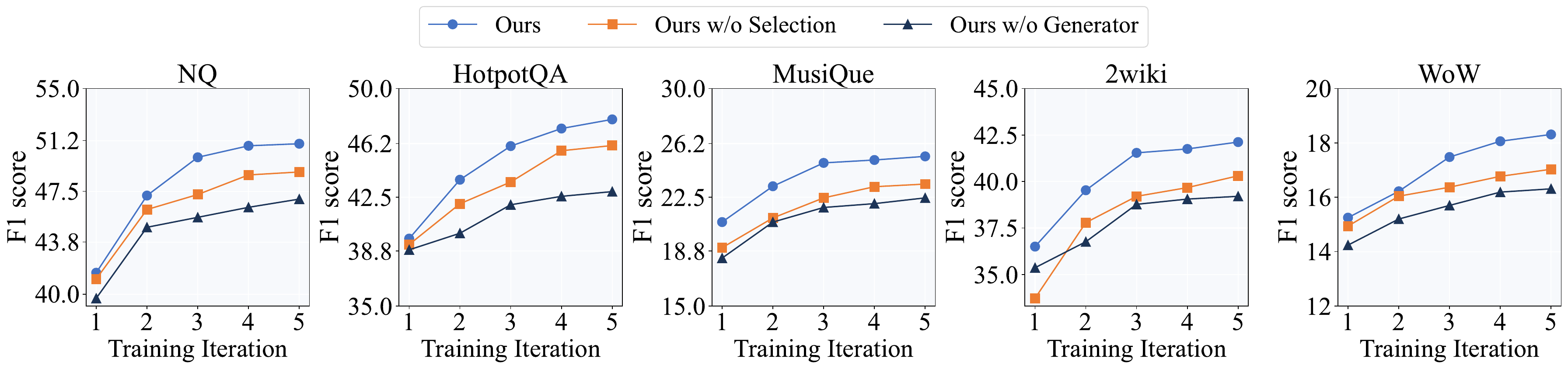}
        \caption{Ablation study on five datasets to demonstrate the effectiveness of training the selection model $\theta_s$ and generator $\theta_g$.}\label{fig:ablation}
\end{figure*}

\begin{figure}[!t]
    \centering
	\includegraphics[width=0.9\linewidth]{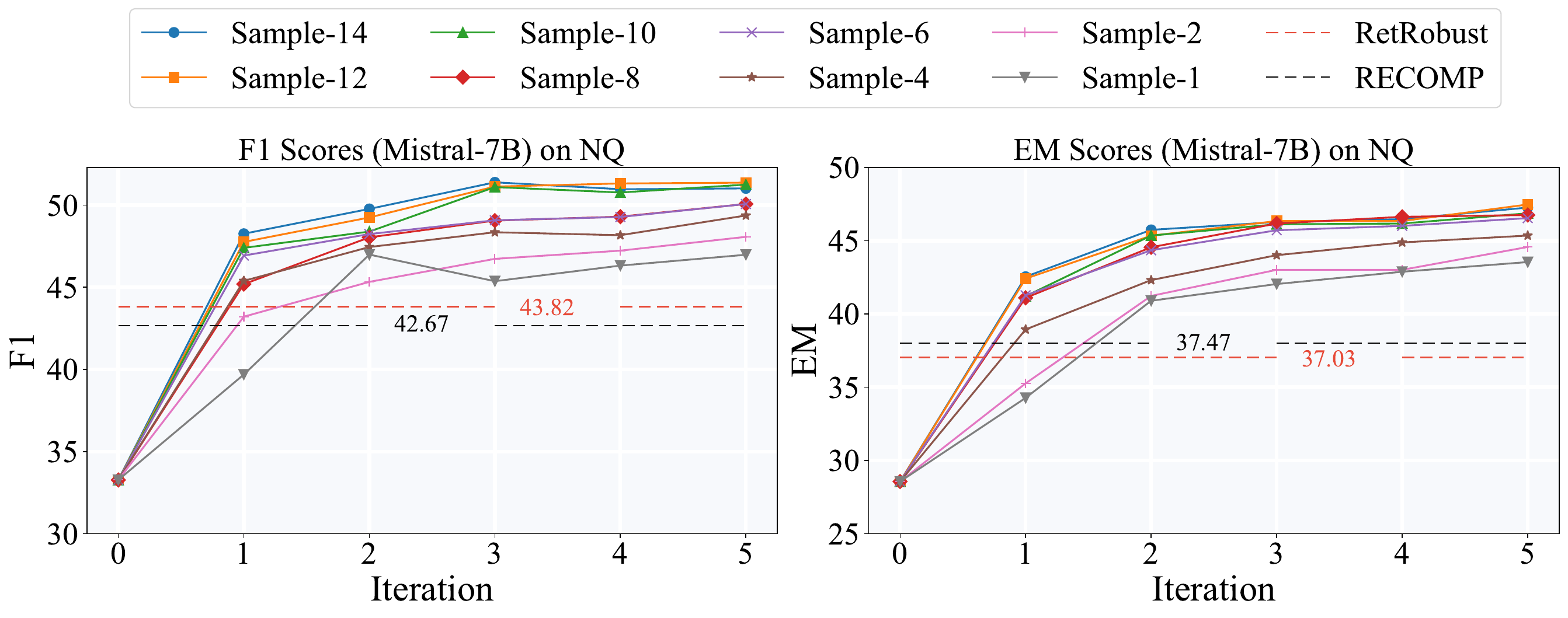}
    \caption{Performance under different sampling number (\S~\ref{sec:sample-number}).}
    \label{fig:sampling}
\end{figure}

\subsection{Human evaluation}
Considering the potential bias of automatic metrics~\citep{Shi2023RADERD}, we conduct a human evaluation with three educated individuals assessing the \textit{correctness} of 100 randomly sampled cases from five benchmarks, using a three-point scale. Each query is paired with the corresponding golden documents and ground truth answers from the original datasets, which serve as references for the human evaluators. 
We ask at least two annotators to evaluate the same case repeatedly. If there is a discrepancy between two annotators, and ask a third annotator to recheck it.
The results are presented in Table~\ref{tab:human}.
The \ours achieves a correctness score of 0.41, while strong open-source baselines only score between 0.32 and 0.37, demonstrating the advantage of our proposed method. The average Kappa statistic for our human evaluation is 0.751, indicating strong agreement.
\begin{table}[htbp]
\centering
\small
\caption{Human evaluation on 100 randomly sampled cases.}
\begin{adjustbox}{width=0.75\columnwidth,center}
\setlength\tabcolsep{7pt}
\begin{tabular}{@{}c cccc@{}}
\toprule
 & \textbf{DPA-RAG} & \textbf{RetRobust} & \textbf{DDR-RAG}  & \textbf{\ours-Mistral} \\
 \midrule
\textbf{Correctness} & 37/100  & 32/100 & 33/100   & 41/100 \\ 
\bottomrule
\end{tabular}
\label{tab:human}
\end{adjustbox}
\end{table}

\section{Discussion}\label{sec:discussion}

\subsection{Impact of the Sampling Number}\label{sec:sample-number}

In our training procedure, for each input query, we sample $m = 8$ document permutations $z$ from the selection model to construct the ELBO objective via importance sampling. To further examine the effect of this sampling number on training dynamics and model performance, we vary $m$ across $\{1, 2, 4, 6, 8, 10, 12, 14\}$ and evaluate the resulting models under the same experimental setup described in Table~\ref{tab:main}.
Figure~\ref{fig:sampling} summarizes the results. We derive several observations from the analysis:
\begin{enumerate*}[label=(\roman*)]
    \item \textbf{Higher sampling improves performance.} Overall, increasing the sampling number $m$ consistently improves the final F1 score. This trend suggests that using more samples provides better approximation of the expected objective, thereby guiding the optimization process more accurately.
    \item \textbf{Sampling number affects convergence speed.} Larger values of $m$ not only lead to better performance but also accelerate convergence. For example, when $m = 14$, the model reaches peak performance within just two training iterations. In contrast, with only $m = 1$, the model requires up to four iterations to achieve similar results. This aligns with our variance analysis in Section~\ref{sec:variance}, as increased sampling reduces training variance and facilitates faster optimization.
    \item \textbf{Trade-off between cost and effectiveness.} While higher values of $m$ yield better and faster results, they also incur greater computational overhead. In practice, we find that using $m = 4$ strikes a good balance, offering competitive performance with reasonable training efficiency.
\end{enumerate*}
These findings highlight the importance of sampling strategies in importance-weighted optimization and suggest practical guidelines for choosing sampling numbers based on resource availability and convergence requirements.

% \subsection{Impact of the sampling number}\label{sec:sample-number}
% In our experiment, for each input query in the training data, we sampling 8 document permutation $z$ from selection model.
% To further investigate the impact of the sampling number, we vary $m$ from $\{1, 2, 4, 6, 8, 10, 12, 14\}$ and evaluate training performance, respectively, under the same setting as in Table~\ref{tab:main}.
% Figure~\ref{fig:sampling} presents the results.
% Our findings indicate the following:
% \begin{enumerate*}[label=(\roman*)]
%     \item As a general principle, a higher sampling number improves performance.
%     \item Increased sampling accelerates convergence. For example, with $m=14$, the model converges within 2 training iterations, whereas with $m=1$, convergence requires 4 iterations.
%     \item In practice, sampling four times ($m=4$) achieves balance between computational efficiency and training effectiveness.
% \end{enumerate*}

\begin{figure}[htbp]
    \centering
	\includegraphics[width=0.9\linewidth]{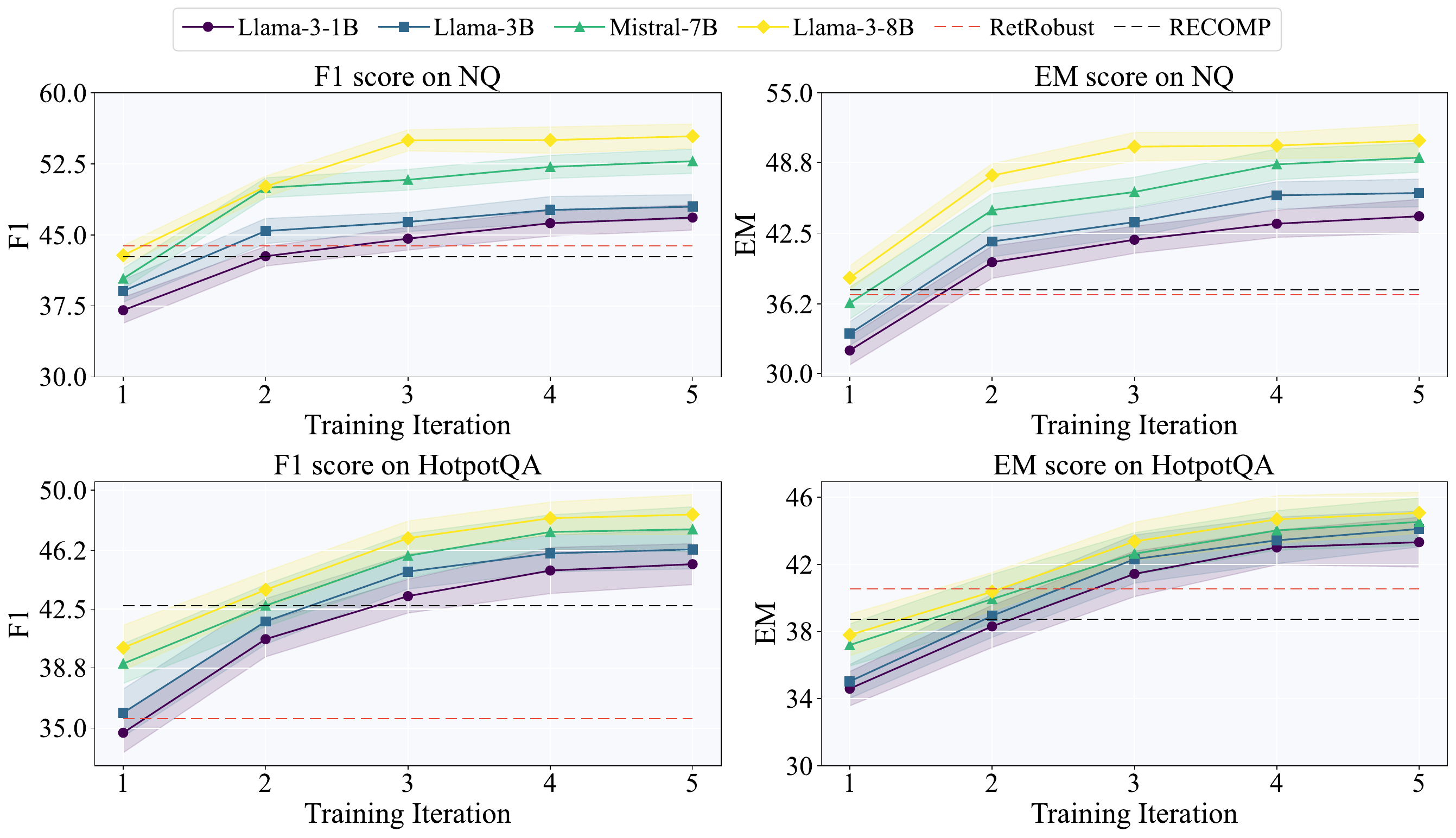}
    \caption{Performance with different parameter sizes of selection model $\theta_s$.}
    \label{fig:scaling}
\end{figure}

\subsection{Performance with scaling-down parameter}
In our experiment, we follow prior work (e.g., RankVicuna~\cite{pradeep2023rankvicuna}) and employ a general-purpose LLM (e.g., Llama) as the backbone for list-wise selection.
Recent studies typically scale up parameter sizes to explore performance upper bounds.
However, since the document selection task involves long context sequences as input with a concern of increased latency, we investigate a more practical low bound by scaling down the size of selection model $\theta_s$ in \ours.
Specifically, we implement our method using Llama-3-8B as the generator and pairing it with smaller LLMs as selection models.
We evaluate the performance under the same conditions as Table~\ref{tab:main} and present the results in Figure~\ref{fig:scaling}.
We observe that as the parameter size of model $\theta_s$ increases from 1B to 8B, performance improves substantially.
The Llama-3-8B generator, when paired with the smallest selector (Llama-3-1B), outperforms strong baselines such as RetRobust, pushing the F1 score from 43.82 to 46.83 on NQ dataset.
Additionally, as an empirical suggestion, training a 3B model (e.g., Llama-3-3B) as the selection model offers a balanced trade-off between performance and computational cost in our method.

\subsection{\todo{Case study}}\label{sec:case}
We conduct case studies to intuitively analyze the advantages and disadvantages of \ours.
Below, we first show the system prompt to the document selection model $\theta_s$ and answer generation model $\theta_g$.
Then, we present the correctly completed cases and the bad cases, respectively, for comparison.

\header{System prompt for document selection models}
In the prompt for the selection model $\theta_s$, we instruct the model to generate a ranked list of useful documents in a generative, auto-regressive manner, where document identifiers are produced in descending order of utility. Documents that contain the correct answer are considered more useful and are expected to be ranked higher.

\begin{lstlisting}[basicstyle=\small\ttfamily, breaklines=true, breakindent=0em, commentstyle=\color{red!50!green!50!blue!50}, frame=shadowbox, rulesepcolor=\color{red!20!green!20!blue!20},numbers=none,literate={`}{\textasciigrave}1]
You are RankLLM, an intelligent assistant that can rank passages based on their relevance and usefulness to the user's query.

I will provide you with {n_docs} passages. Please rank these passages based on their usefulness in answering the user's search query: "{question}". 

A passage's usefulness is defined as:
1. Relevance to the question.
2. Contains necessary information to help the user.

The passages are listed below, each with a numerical identifier [].
{docs}

Rank the {n_docs} passages based on their usefulness in descending order. Use the format [] > [], e.g., [2] > [1]. Only respond with the ranking results; do not provide explanations.

Search Query: {question}
Your output:
\end{lstlisting}

\header{System prompt for answer generation model}
In the prompt for the answer generation model $\theta_g$, we provide the documents selected by the selection model as references, and instruct the model to generate the final answer grounded in the information contained within these documents.

\begin{lstlisting}[basicstyle=\small\ttfamily, breaklines=true, breakindent=0em, commentstyle=\color{red!50!green!50!blue!50}, frame=shadowbox, rulesepcolor=\color{red!20!green!20!blue!20},numbers=none,literate={`}{\textasciigrave}1]
You are an artificial intelligence assistant. You should gives helpful and precise answers to the user's questions based on the context. The context here refer to external passages related to user's questions.

Please answer the question: "{question}" using provided passages. Each passage is indicated by a numerical identifier [].
Here are the related passages for your reference.
{docs}
Question: {question}
Your answer: 
\end{lstlisting}

\header{Correctly completed cases}
Below, we provide a case in which the input query asks: \textit{Which games, Strange Synergy or Qwirkle, is a card game published by Steve Jackson Games?}, with the ground truth being \textit{Strange Synergy}.
To answer this question, it is essential to retrieve not only detailed descriptions of Strange Synergy but also background information about Qwirkle for comparative analysis.

Among the initially retrieved documents, document \texttt{[1]}, \texttt{[2]}, and \texttt{[17]} are all individually necessary to resolve the query.
Specifically, passages \texttt{[1]} and \texttt{[2]} provide overlapping but essential descriptions of Strange Synergy, consistently identifying it as a card game published by Steve Jackson Games. 
Meanwhile, passage \texttt{[17]}, though it offers necessary contrastive information about Qwirkle, ranked much lower initially.
In \ours, the selection model successfully identifies \texttt{[1]}, \texttt{[2]}, and \texttt{[17]} as the top-ranked passages, demonstrating its ability to select  the most relevant target information. Based on these selected passages, the generator compares the publisher information from \texttt{[1]} and \texttt{[17]}, accurately associating Strange Synergy with Steve Jackson Games as the correct answer.
This example highlights the strength of our \ours in enabling the cooperation between the selection and generation models.

\begin{lstlisting}[basicstyle=\small\ttfamily, breaklines=true, breakindent=0em, commentstyle=\color{red!50!green!50!blue!50}, frame=shadowbox, rulesepcolor=\color{red!20!green!20!blue!20},numbers=none,literate={`}{\textasciigrave}1]
INPUT QUERY: Which games, Strange Synergy or Qwirkle,  is a card game published by Steve Jackson Games? (Ground truth: Strange Synergy)

INITIALLY RETRIEVED DOCUMENTS: [0] Title: Strange Synergy. Content: Strange Synergy Strange Synergy is a card game published by Steve Jackson Games in which players build a team of super heroes to battle an opponent's team. The game supports up to four players... [1] Title: Strange Synergy. Content: enemies for three turns in a row, or capture an opposing team's flag. Strange Synergy Strange Synergy is a card game published by Steve Jackson Games in which players build a team of super heroes to battle an opponent's team. The game supports up to four players. However, there are bases and flags for two more players. [2] ... [3] ... [4] ... [5] ... [6] ... [7] ... [8] ... [9] ... [10] ... [11] ... [12] ... [13] ... [14] ... [15] ... [16]... [17] Title: Qwirkle. Content: Qwirkle Qwirkle is a tile-based game for 2 to 4 players, designed by Susan McKinley Ross and published by MindWare. Qwirkle shares some characteristics with the games Rummikub and Scrabble. It is distributed in Canada by game and puzzle company, Outset Media. Qwirkle is considered by MindWare to be its most awarded game of all time. [18] ... [19] ...

OUTPUT OF SELECTION MODEL: [1] > [2] > [17]
OUTPUT OF GENERATOR MODEL: Strange Synergy
\end{lstlisting}

\header{Bad case of \ours}
\todo{We also investigate the bad case of \ours to analyze potential disadvantages of our method.
The Table~\ref{tab:bad_cases} summarize two key type of errors: incorrect selection and generation mismatch.
First, the \textit{incorrect selection} indicates that the selector fails to comprehensively include all  documents that are necessary for answer the question.
Without such contrastive evidence, the generator tends to produce incorrect answers, even when relevant documents are present in the retrieved pool.
This suggests the need for diversity-aware selection strategies that explicitly promote comparative reasoning.
Second, the \textit{generation mismatch} indicates that the generator model produce a incorrect answer though the ground truth documents have been incorporated into the context.
This often occurs when the input contains partially relevant or distracting content, i.e., the noise.
These cases highlight the challenge of robust generation.}

% These observations underline that the effectiveness of DRO depends not only on selection accuracy, but also on the alignment between selected evidence and generation behavior.

\begin{table}[htbp]
\centering
\caption{Representative failure cases of DRO. We summarize two typical types of failure, i.e., selection and generation mismatch, and suggest potential remedies.}

\label{tab:bad_cases}
\begin{adjustbox}{width=\columnwidth,center}
\setlength\tabcolsep{14pt}
\begin{tabular}{p{3.2cm} | p{3.5cm} | p{3.5cm} | p{2.8cm}}
\toprule
\textbf{Case Type} & \textbf{Description} & \textbf{Observed Behavior} & \textbf{Potential Remedies} \\
\midrule
Selector misses contrastive document & 
The selector fails to comprehensively include documents (e.g., different publishers), which are necessary for disambiguation. & 
The generator outputs an incorrect answer due to the lack of comparative evidence. & 
Incorporate multi-view selection or contrastive supervision to enforce diversity in document selection. \\
\midrule
Generator is misled by noisy content & 
Although the selector provides useful documents, the generator fail to generate the correct answer. & 
The LLM misunderstands the input documents and predicts incorrect answer & 
Introduce answer verification module chain-of-thought reasoning process for LLM before generating the final answer. \\
\bottomrule
\end{tabular}
\end{adjustbox}
\end{table}

\section{Conclusion}\label{sec:conclusion}

We have presented \ours, a direct retrieval-augmented optimization method for the RAG task that (i) treats document permutations of the retriever as latent variables, and (ii) enables the co-training of a list-wise document selection model and an LLM generator through an expectation-maximization approach.  
Specifically, \ours alternates between two steps:  
(1) directly estimating the distribution of document permutations from the selection model using an importance sampling strategy; and  
(2) maximizing the importance-weighted evidence lower bound to jointly train the selection model and the LLM generator.
Through theoretical analysis, we have proven that the learning objectives in \ours are analogous to policy-gradient reinforcement learning, reinforcing the selection and generation models with an end-to-end training reward.
Extensive experiments conducted on five datasets have validated the superiority of our \ours.
For future work, we aim to extend this approach to multi-modal RAG scenarios.  
We also plan to explore the co-training of additional retrieval modules within \ours, such as a query re-writer.

\clearpage
\newpage
\bibliographystyle{ACM-Reference-Format}
\bibliography{references}

\end{document}